\documentclass[draft]{agujournal2019}
\usepackage{url} 
\usepackage{lineno}
\usepackage[inline]{trackchanges} 
\usepackage{soul}
\usepackage{amsmath}
\usepackage{amssymb}
\usepackage{amsfonts}
\usepackage{bbm}
\usepackage{tikz}
\usepackage{float}
\usepackage{enumitem}
\usetikzlibrary{positioning} 
\usetikzlibrary{calc}
\usepackage{amsthm} 
\newtheorem{theorem}{Theorem}
\theoremstyle{remark}
\newtheorem{remark}{Remark}
\draftfalse

\journalname{Journal of Advances in Modeling Earth Systems (JAMES)}

\begin{document}

%
%

\title{Distillation and Interpretability of Ensemble Forecasts of ENSO Phase using Entropic Learning}

%
%




\authors{Michael Groom\affil{1}, Davide Bassetti\affil{2}, Illia Horenko\affil{2}, Terence J. O'Kane\affil{3}}

\affiliation{1}{CSIRO Environment, Eveleigh, New South Wales, Australia}
\affiliation{2}{Faculty of Mathematics, Rheinland-Pfälzische Technische Universität Kaiserslautern Landau, Kaiserslautern, Germany}
\affiliation{3}{CSIRO Environment, Hobart, Tasmania, Australia}




\correspondingauthor{Michael Groom}{Michael.Groom@csiro.au}



\begin{keypoints}
\item We compress large entropic learning ensembles into compact distilled models while maintaining skill and improving interpretability
\item Diagnostic maps derived from the distilled models successfully trace the evolution of ENSO events from their initial precursor states
\item Important predictors are linked to known ocean-atmosphere precursors, highlighting physical consistency and improving trustworthiness 
\end{keypoints}

%
%

%
%


\begin{abstract}
This paper introduces a distillation framework for an ensemble of entropy-optimal Sparse Probabilistic Approximation (eSPA) models, trained exclusively on satellite-era observational and reanalysis data to predict ENSO phase up to 24 months in advance. While eSPA ensembles yield state-of-the-art forecast skill, they are harder to interpret than individual eSPA models. We show how to compress the ensemble into a compact set of ``distilled'' models by aggregating the structure of only those ensemble members that make correct predictions. This process yields a single, diagnostically tractable model for each forecast lead time that preserves forecast performance while also enabling diagnostics that are impractical to implement on the full ensemble. 

An analysis of the regime persistence of the distilled model ``superclusters'', as well as cross-lead clustering consistency, shows that the discretised system accurately captures the spatiotemporal dynamics of ENSO. By considering the effective dimension of the feature importance vectors, the complexity of the input space required for correct ENSO phase prediction is shown to peak when forecasts must cross the boreal spring predictability barrier. Spatial importance maps derived from the feature importance vectors are introduced to identify where predictive information resides in each field and are shown to include known physical precursors at certain lead times. Case studies of key events are also presented, showing how fields reconstructed from distilled model centroids trace the evolution from extratropical and inter-basin precursors to the mature ENSO state. Overall, the distillation framework enables a rigorous investigation of long-range ENSO predictability that complements real-time data-driven operational forecasts.
\end{abstract}

\section*{Plain Language Summary}
El Ni\~no--Southern Oscillation (ENSO) is the dominant mode of interannual climate variability in the Pacific that alternates between El Ni\~no, La Ni\~na, and neutral phases and has broad societal and economic impacts. State-of-the-art methods are now able to reliably forecast the phase of ENSO up to 2 years in advance, often by combining the predictions of many individual models into a large ensemble. While ensembling improves forecasting skill, it also makes it harder to interpret why a particular forecast was made. In this study, we develop an approach for compressing a large ensemble of models into a much smaller set of ``distilled'' models that are still able to make accurate long-range forecasts but are also much easier to interpret. Our work provides a practical tool that makes advanced forecasting techniques both trustworthy and useful, not just for issuing forecasts, but for deepening our scientific understanding of what makes the climate predictable.

\section{Introduction} \label{sec:intro}

The El Ni\~no--Southern Oscillation (ENSO) is the dominant mode of interannual climate variability, and accurate seasonal-to-interannual forecasts of ENSO are critical for mitigating climate-related risks across global ecosystems and economies. While dynamical models based on coupled ocean-atmosphere physics have historically served as the operational standard, they continue to face challenges related to model bias and the well-known spring predictability barrier \cite{Barnston2012, Tippett2012}. In recent years, data-driven approaches (particularly those utilising Deep Learning) have emerged as competitive alternatives, demonstrating the ability to exceed the forecast skill of dynamical models and remain skilful at lead times of up to 18 months or longer \cite{Ham2019,Zhou2023}. However, the intrinsically ``black box'' nature of Deep Learning (DL) models raises concerns regarding the physical plausibility and trustworthiness of their predictions.

In an attempt to address these concerns, a growing body of research has sought to integrate interpretability and Explainable AI (XAI) methods, into DL-based ENSO forecasting. For example, \citeA{Wang2023a} demonstrated that their spatio-temporal information extraction and fusion DL model has skill out to 22 months, and used backpropagation to derive the gradient of the prediction with respect to each input variable at each grid point as a means of extracting sources of ENSO predictability. \citeA{Chen2025} and \citeA{Lyu2024} developed hybrid architectures combining Convolutional Neural Networks (CNNs) and Transformers that achieve correlation skill exceeding 0.5 at lead times of up to 20 months. \citeA{Chen2025} applied a gradient-based sensitivity analysis to show that their model captures physical ENSO precursors that are consistent with established mechanisms, as well as potentially offering new insights on inter-basin interactions at longer lead times, while \citeA{Lyu2024} used Integrated Gradients \cite{Sundararajan017} to show which regions of sea surface temperature (SST) anomalies are most sensitive for predicting the Ni\~no3.4 index at varying lead times for El Ni\~no vs. La Ni\~na events. 

Although the XAI methods employed in these studies are primarily external methods for adding interpretability to a given DL model---rather than being ``baked-in'' to the model structure itself---analysing the predictions made by DL models using these interpretability techniques has often led to the verification of known physical precursors, and in some cases has highlighted mechanisms that were previously unknown or under-appreciated. \citeA{Zhou2024} performed perturbation experiments on their Transformer-based model to quantify input-output relationships, revealing the existence of upper-ocean temperature anomaly pathways that propagate in a counter-clockwise loop consisting of two bands at 5$^\circ$S-5$^\circ$N and 8$^\circ$N-15$^\circ$N. \citeA{Li2024} quantified the relative contributions of SST anomalies in different ocean basins, confirming the dominant role of the tropical Pacific as well as highlighting the role of remote forcing from the Indian and tropical Atlantic oceans at longer lead times. \citeA{Wang2024a} demonstrated the important role that sea surface salinity plays in improving skill at lead times greater than 6 months, while \citeA{Colfescu2024} showed that 10m zonal wind, in particular over the Indian Ocean, has comparable predictive skill to that of SST for lead times between 11 and 21 months. When combined with linear correlation/regression analysis, this finding was suggested to be mechanistically linked to preceding SST anomalies over the western Indian Ocean and their subsequent eastward propagation via the 10m wind anomalies. Finally, \citeA{Chen2025a} identified a tropical Pacific Ocean mode in 300m heat content that is responsible for much of the skill beyond 1 year in the popular CNN model of \citeA{Ham2019}.

DL models have also proven capable of capturing ENSO diversity and asymmetry. \citeA{Sun2023} developed a CNN-based model for predicting SST anomalies in 13 different regions across the tropical Pacific, with skill out to 12 months, thereby being able to distinguish between Central and Eastern Pacific events. The authors also generated activation maps for successfully predicted events, finding that precursors in the north Pacific, south Pacific and tropical Atlantic play a critical role at 10 months lead time, while for correctly predicted events at 16 months lead time, signals in the tropical Pacific associated with the discharge-recharge cycle are most important. \citeA{RiveraTello2023} also utilised a CNN-based model to predict both Eastern (EP) and Central (CP) Pacific ENSO diversity indices, with a focus on strong EP El Ni\~no events. By employing the technique of Layerwise Relevance Propagation \cite{Binder2016}, they were able to show that well-known precursors, such as high SST and sea surface height (SSH) anomalies in the eastern Pacific in combination with westerly wind anomalies in the central-western equatorial Pacific, were most relevant for successfully predicting events where the EP index was greater than 1.5, as well as some (perhaps) less well-known precursors such as southerly wind anomalies in the north-eastern Pacific, consistent with the wind signature of the positive phase of the North Pacific Meridional Mode \cite{Chiang2004}. 

Despite these advances, some challenges remain. For example, gradient-based explanations are post-hoc approximations that can be unstable or noisy, \textbf{raising questions about whether they truly represent the model's decision-making process} \cite{Adebayo2018,Ghorbani2019}. Furthermore, the vast majority of the studies cited above use Coupled Model Intercomparison Project (CMIP) climate model output as their main source of training data. This is often due to there being insufficient observational data required for training performant DL models (for example, the \citeA{Ham2019} CNN model requires at least $\mathcal O(1000)$ samples according to \citeA{Chen2025a}), but risks the introduction of known systematic biases into the DL model outputs \cite{Zhou2023}. \textbf{It is therefore highly desirable to be able to replicate the findings of these studies while relying solely on high-quality observational/reanalysis products}, e.g. datasets produced using satellite observations and other modern remote sensing methods. This necessitates the use of a methodology other than deep learning for circumventing the ``small data'' problem \cite{Horenko2020}. 

Entropic Learning (EL) methods such as the entropy-optimal Sparse Probabilistic Approximation (eSPA) algorithm, originating from the pioneering work by \citeA{Gerber2020}, are explicitly designed for this small data regime and have been shown to provide predictions that are comparable to or better than DL across a wide range of tasks, but requiring vastly fewer model parameters and at a small fraction of the computational cost \cite{Horenko2020,Vecchi2022,Horenko2023,Vecchi2024,Bassetti2024,Bassetti2025}. By restricting model complexity through information-theoretic principles, \textbf{EL allows for models to be trained directly on the limited historical observational record without needing to rely on climate model outputs}, while also obtaining state-of-the-art skill in long-range ENSO prediction \cite{Groom2024,Groom2025}. Another key advantage of EL models is their inherent interpretability; rather than requiring post-hoc analysis, \textbf{the model structure itself---based on clustering in a probabilistically sparsified feature space---is completely transparent}. This structure enables insights that are otherwise difficult to obtain for the class of problems considered here. For example, the feature importance vectors that are learnt by EL models are not possible to obtain using common techniques such as permutation importance \cite{Altmann2010} for the types of learning tasks considered in this paper, due to the existence of serial correlations in the data.

In \citeA{Groom2025}, we sought to operationalise our earlier results on ENSO prediction with eSPA \cite{Groom2024} by generating an ensemble of eSPA models and assessing skill against the International Research Institute (IRI) ENSO prediction plume \cite{Ehsan2024}. Compared to the IRI plume, the eSPA ensemble produced probabilistic forecasts with comparable skill over the range of lead times for which the IRI forecasts are issued for, while maintaining skill out to 20-24 months lead time (depending on the particular skill metric). However, this performance came at a cost, since the large ensemble of models made applying the interpretability techniques introduced in \citeA{Groom2024} difficult, as they are intended to be applied to a single eSPA model. This trade-off between ensemble performance and model interpretability motivates the application of \textit{distillation} techniques. \citeA{Bucilua2006} introduced the concept of model compression, demonstrating that the function learned by a large ensemble of classifiers could be compressed into a single, smaller model without significant loss of accuracy. This idea was generalised by \citeA{Hinton2015} to knowledge distillation, whereby a compact ``student'' model is trained to mimic the soft targets (i.e. class probabilities) of a more cumbersome ``teacher'' ensemble. 

In this study, we invoke the term distillation to denote a similar goal as those studies, i.e. to learn a more compact model from a large ensemble while retaining as much performance as possible, while noting that the specific methodology introduced here for achieving this goal in the context of EL is different. Specifically, we apply distillation to an eSPA ensemble forecasting system very similar to the one introduced in \citeA{Groom2025}. Rather than training a neural network student, we leverage eSPA's clustering structure to aggregate the internal states of successful ensemble members into a reduced set of ``superclusters.'' This process distils the ensemble into a single, parsimonious representation for each lead time, retaining the probabilistic skill of the original ensemble while enabling a rigorous examination of how successful long-range ENSO predictions are made. In addition to the various interpretability techniques introduced in \citeA{Groom2024}, which can be applied to the distilled models with only minor changes, we introduce new techniques that are analogous to the gradient-based sensitivity/salience methods for DL but which avoid many of the known pitfalls of those methods.

The remainder of this paper is organised as follows. Section \ref{sec:methods} describes the datasets, pre-processing, ensemble forecast procedure, distillation methodology as well as various quantities that can be derived from the distilled models. Section \ref{sec:results} presents the main results, including feature importance visualisation, case studies of key events and the identification of precursors. Section \ref{sec:conclusions} discusses the implications of the findings and concludes the paper. In addition to the figures shown in this paper and the supporting information, a full suite of interactive visualisations for every plot type and combination of lead time, class etc is available via a web app hosted at \url{https://app.michaelgroom.net/}. The interested reader is highly encouraged to explore these interactive visualisations, as they allow for a much richer examination of the full set of results than can be conveyed through the limited number of static images in this manuscript.

\section{Methodology} \label{sec:methods}

\subsection{Datasets and pre-processing} \label{subsec:pre-processing}
Following \citeA{Groom2025}, only observational/reanalysis data from the satellite era are employed when training and validating the entropic learning models in this study. A similar set of oceanic and atmospheric fields, as in \citeA{Groom2025}, is used as predictors, with only minor updates to how these fields are processed, which will be summarised here. 

The oceanic fields considered are monthly means of global sea surface temperature (SST) between 60$^\circ$S-60$^\circ$N and the vertical derivative of subsurface temperature ($\mathrm dT/\mathrm d z$) between 30$^\circ$S-30$^\circ$N and restricted to longitudes of 120$^\circ$E-90$^\circ$W and depths of 0-700m. The atmospheric fields considered are monthly means of the zonal and meridional surface wind stresses ($\tau_x$ and $\tau_y$), restricted to latitudes of 20$^\circ$S-20$^\circ$N and longitudes of 120$^\circ$E-80$^\circ$W, with a mask applied so that only oceanic wind stresses are selected. The SST and wind stress data are taken from the ERA5 reanalysis \cite{Hersbach2020} on a global 0.25$^\circ\times$0.25$^\circ$ grid. SST comes from the boundary conditions of ERA5, which are provided by the HadISST2.1.1.0 product \cite{Titchner2014} prior to September 2007 and OSTIA \cite{Donlon2012} thereafter. The $\mathrm dT/\mathrm d z$ data are calculated from the 3D potential temperature field provided by the ORAS5 reanalysis \cite{Zuo2019}, which is forced by CORE bulk formulas \cite{Large2009} derived from ERA-Interim \cite{Dee2011} from 1979 to 2014 and from the ECMWF operational NWP system thereafter. All three datasets are restricted to the satellite era (taken to be January 1979 onwards) so as to avoid the much higher uncertainties in the reanalysis products prior to them assimilating satellite observations. 

The Ni\~no3.4 index is calculated from the ERA5 SST field using a 1991-2020 climatology for consistency with real-time operational forecasts \cite{Ehsan2024}. Unlike in \citeA{Groom2025}, in this study we use the relative Ni\~no3.4 index \cite{LHeureux2024}, whereby the tropical mean anomaly is subtracted from the Ni\~no3.4 index and the resulting index is rescaled to have the same variance as the original index. This has the advantage of removing any net warming trend from the Ni\~no3.4 index and has also been demonstrated to be less sensitive to the particular choice of climatological period. The 1991-2020 climatology is also used for calculating anomalies of the SST, $\mathrm dT/\mathrm d z$ and wind stress fields, which are then linearly detrended. The anomalies are also regridded to a 1$^\circ\times$1$^\circ$ grid using bilinear interpolation, with an isotropic Gaussian filter applied for anti-aliasing. 

Following \citeA{Groom2025}, a Singular Spectrum Analysis (SSA) is performed on each of the three fields (the two wind stress components are stacked together in the SVD) using a 12-month embedding length. Initially, a fixed percentage of the total variance in each field (85\% for SST and $\mathrm dT/\mathrm d z$ and 70\% for wind stress) is used to select the number of PCs to retain as features (162, 188 and 83 respectively), which is then further trimmed using the method described in section \ref{subsec:importance} below. To improve skill at short lead times, the monthly Ni\~no3.4 index is also added as a feature, along with the class probability distributions at lead times up to $n-1$ months \cite{Groom2025}. The targets for prediction (class probability distributions) are generated by considering the probability of the Ni\~no3.4 index being greater than $0.5^\circ$C (El Ni\~{n}o), less than $-0.5^\circ$C (La Ni\~{n}a) or neither (neutral) in $n$ months time. As in \citeA{Groom2025}, the 3-month running average of the Ni\~no3.4 index is used for defining the classes, which are labelled as classes -1, 0 and 1 respectively when calculating metrics that depend on the ordinal ranking of classes, such as the ranked probability score. As of January 2025, a total of 552 instances are available for training.

\subsection{Ensemble forecast procedure} \label{subsec:forecast-procedure}

Initially, a set of predictions are generated by training an ensemble of eSPA classifiers at each lead time $n$ and then averaging over their individual predictions. Note that separate eSPA models are trained for each lead time of 1, 2, 3, $\ldots$, 24 months, as opposed to a single model that makes predictions for multiple lead times, and for each new forecast (i.e. a sequence of 1-24 month lead predictions starting from a particular date) a completely independent set of models is trained from scratch. The ensemble forecasts are performed using the same methodology described in \citeA{Groom2025}, which will be briefly summarised here. 

For each year and month in the evaluation period (which contains start dates from January 2002 to December 2022), a feature matrix $X$ is loaded, containing only instances up to the month prior to the start date (i.e. the first month to be predicted). A final pre-processing step of mapping the data to a uniform distribution with values between 0 and 1 via a quantile transformation is applied to each of the features in $X$. Then, for each lead time $n$ from 1 to 24 months, the class probability matrix $\Pi^{(n)}$ is loaded, containing only instances up to $n+1$ months prior to the start date (reflecting real-time conditions where the data can only be labelled up to the month at which the forecast is issued). To help condition the predictions of each model at lead time $n$ on the sequence of class probabilities that have already been observed/predicted, we provide as additional features the class probability distributions up to $n-1$ months. For the training set these will be the true distributions, while for the prediction at $n$ months ahead of the latest available data we provide the mean predictions from the ensemble that have already been made up to $n-1$ months ahead. Furthermore, to help with the seasonal variability in ENSO precursors, only instances within $\pm1$ month of the target month are retained. 

Given the data $(X,\:\Pi^{(n)})$, an ensemble of 50 eSPA models is trained by performing 50 separate train/test splits of the data (keeping 80\% of the data from training and 20\% of the data for testing), conducting a grid search for the optimal hyperparameters ($K$, $\varepsilon_E$ and $\varepsilon_C$ being the number of clusters, the entropy weighting and classification weighting respectively) on each split and saving the model with the best score on the test set (measured using the ranked probability score). These 50 models are then used to predict the class probability distribution for the most recent instance in $X$ (which is unlabelled and therefore not in the training set), with each of these individual predictions averaged to give the ensemble mean prediction at lead $n$ for that forecast. For further details see \citeA{Groom2025}. One important difference from that study is that, rather than perform a new SSA decomposition, updated every year of the hindcast period (thus reflecting real-time conditions), here we perform a single SSA using the entire dataset to generate the PCs used in the feature matrix $X$. While all other anti-leakage protocols from \citeA{Groom2025} are retained, this modification is made to ensure that the same latent space is being used by all eSPA models across all forecasts, which is a necessary assumption to make when applying the various interpretability methods described in the following sections. Comparisons with the strictly real-time approach of \citeA{Groom2025} show that the present approach does not result in any artificial increase in skill (see figure S1 of the supporting information for details).

\begin{figure}[t]
  \centering
  \begin{tikzpicture}
    \node[draw, rounded corners=5pt, line width=1pt, inner sep=3pt] {
      \includegraphics[width=0.95\textwidth]{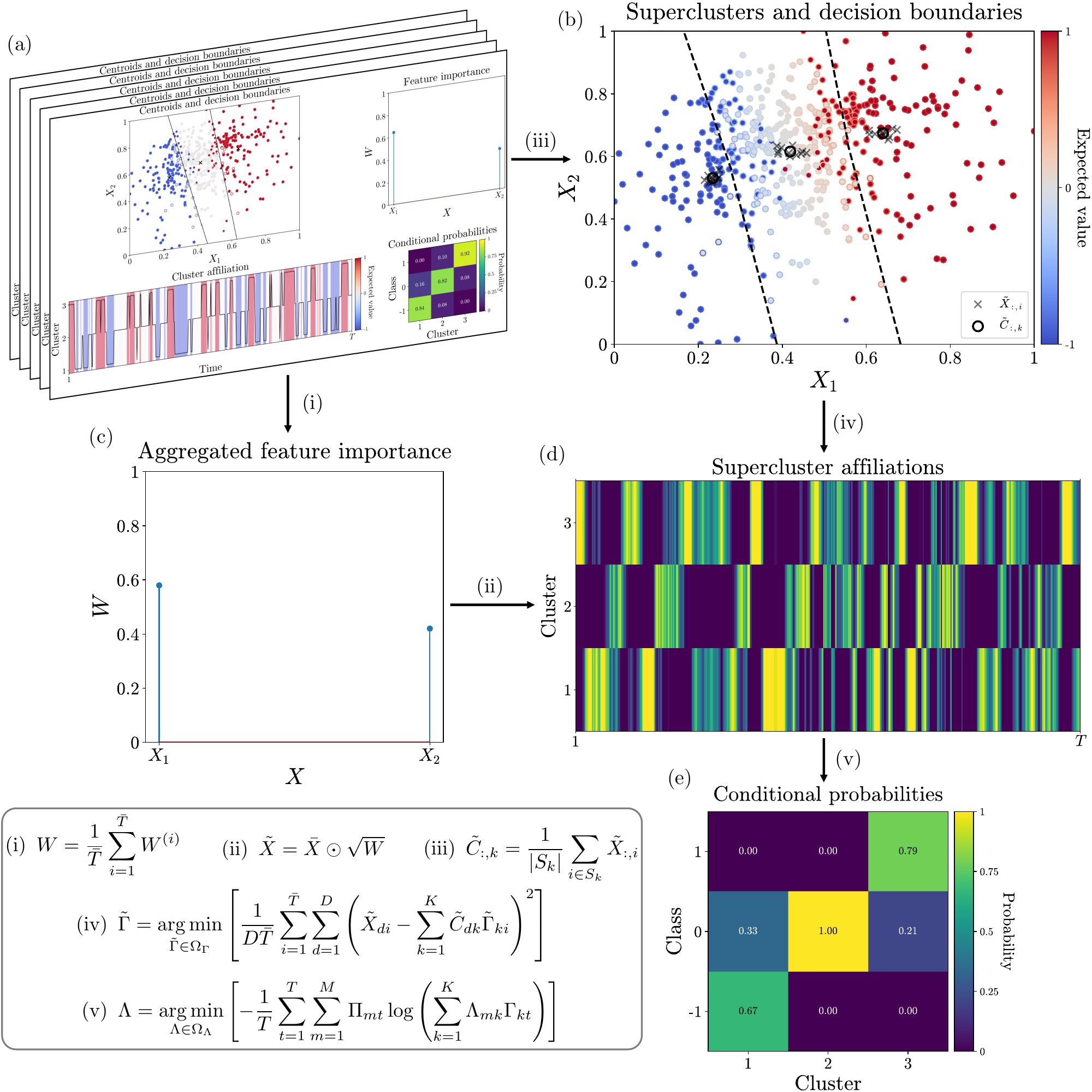}
    };
  \end{tikzpicture}
  \caption{Schematic of the distillation procedure. Starting with an ensemble of eSPA models (a)---see Figure 2 of \citeA{Groom2025} for a complete description---a set of superclusters is fit to the centroids of each eSPA model (b), and the aggregated feature importance is calculated (c). Then, fuzzy affiliations over these superclusters are calculated (d), along with a matrix of conditional probabilities based on these fuzzy affiliations (e). The legend in the bottom left contains the key equations involved in each of these steps. The dashed lines in (b) indicate the decision boundaries.}
  \label{fig:summary} 
\end{figure}

\subsection{Distillation of ensemble forecasts} \label{subsec:distillation}

The approach described in section \ref{subsec:forecast-procedure} above, when combined with the post-processing techniques discussed in \citeA{Groom2025}, has been demonstrated to produce forecasts with probabilistic skill that is comparable to operational baselines (such as the International Research Institute for Climate and Society ENSO prediction plume) over the range of lead times they issue forecasts for (typically $\le$ 12 months), while also being substantially cheaper to operate than a conventional, fully-coupled dynamical seasonal prediction system. However, as discussed in the introduction, the ensemble forecasting approach introduces a lot of additional complexity that makes applying the interpretability methods introduced in \citeA{Groom2024} more challenging, as there is now a total of 302400 models to analyse for the evaluation period of January 2002 to December 2022. This motivates an attempt to distil the information contained in the full ensemble down to a single model for each lead time while retaining as much of the skill as possible, thereby enabling the application of various interpretability methods to this much smaller set of distilled models. A visual summary of this procedure is given in figure \ref{fig:summary}, with a detailed description of each step given in the sections below.

\subsubsection{Affiliations over superclusters} \label{subsec:superclusters}

The distillation process starts by employing a filtering step which is common to many of the DL studies cited in the introduction, namely to select only the models at each lead time from the full ensemble whose predictions for their respective target dates are verified to indicate the correct phase of ENSO. The next step is to collect the centroids from these models that the most recent instance in their respective training set features $X$ was assigned to (which corresponds to the centroid each model relied on to make its prediction in its respective forecast), transform them back into the unscaled feature space through the inverse of the quantile transform that was applied, and save them to a new dataset $\bar X\in\mathbb R^{D\times N\times \bar T}$ that is indexed by feature dimension $d=1,\ldots,D$, lead time $n=1,\ldots,N$ and instance $i=1,\ldots,\bar T$ (note: the maximum possible number of instances at each lead time is $\bar T=12600$). In the schematic in figure 1, these model centroids are displayed as grey crosses in panel (b). The $W$ vector for each correct model is also stored to produce a single, averaged vector $W^{(n)}\in \mathbb R^D$ for each lead time (step (i) and panel (c) of figure \ref{fig:summary}).

Next, $k$-means clustering is performed on the new dataset $\bar X^{(n)}$ for each lead time (step (iii) in figure \ref{fig:summary}), resulting in a set of ``superclusters'' $\tilde C^{(n)}\in\mathbb R^{D\times K}$ representing regions of the feature space that are typically occupied by the centroids of correct eSPA models. In the figure \ref{fig:summary} schematic these superclusters are displayed as black circles in panel (b). To help ensure that the relative importance of each feature that is learned by the individual (correct) eSPA models is preserved, $\bar X^{(n)}$ is scaled by $\sqrt{W^{(n)}}$ (step (ii) in figure \ref{fig:summary}). This is done after after min-max scaling so that every feature dimension has a range of $[0,1]$, reflecting the range over which $W^{(n)}$ was first learnt. This makes $k$-means clustering with the Euclidean metric on $\tilde X^{(n)}=\bar X^{(n)}\odot\sqrt{W^{(n)}}$ equivalent to clustering on $\bar X^{(n)}$ with the weighted-Euclidean metric, i.e. the same metric eSPA uses for clustering but without the additional classification loss term. To choose the number of clusters $K$ to use, initially the gap statistic \cite{Tibshirani2001} was employed for each lead time. However, since the number of clusters across lead times was always observed to be close to a median value of 12, a constant $K=12$ was used for all lead times. 

Once the set of supercluster centroids is obtained, the clustering quality can be further improved by letting the assignment matrix $\tilde\Gamma$ be fuzzy rather than discrete, i.e. $\tilde\Gamma\in\Omega_\Gamma$ where $\Omega_\Gamma=\left\{\Gamma_{kt}\in[0,1]\forall\, k,t:\sum_{k=1}^K\Gamma_{kt}=1\forall\, t\right\}$, and solving the following quadratic program for $\tilde \Gamma^{(n)}$ at each lead time,
\begin{linenomath*}
\begin{equation} \label{eqn:Gamma}
\mathcal{L}\left(\tilde \Gamma^{(n)}\right)=\frac{1}{D\bar T}\sum_{i=1}^{\bar T}\sum_{d=1}^D\left(\tilde X^{(n)}_{di}-\sum_{k=1}^K\tilde C^{(n)}_{dk}\tilde\Gamma^{(n)}_{ki}\right)^2\rightarrow \min_{\tilde\Gamma^{(n)}\in\Omega_\Gamma},
\end{equation} 
\end{linenomath*}
which can be efficiently solved using the spectral projected gradient method described in \citeA{Gerber2020} (step (iv) of figure \ref{fig:summary}). This results in a matrix of fuzzy assignments for each instance in $\tilde X^{(n)}$, however since there can be multiple instances associated with the same target date the assignments $\tilde\Gamma^{(n)}_{:,i}$ are averaged over all instances $i$ corresponding to the same target date $t=1,\ldots,T$ (where $T=276$ for the evaluation period of January 2002 to December 2024) and then renormalised to give a matrix $\Gamma^{(n)}$ of fuzzy cluster assignments over target dates, which we will call the affiliation matrix (shown in panel (d) of figure \ref{fig:summary}). By making the affiliation matrix fuzzy in this way, we are able to better capture the diversity across ENSO events while only relying on a minimal set of superclusters.

Once $\Gamma^{(n)}$ has been calculated, the conditional probability of observing class $m$ given assignment to cluster $k$ on target date $t$ can be calculated by solving the following (convex) nonlinear program for $\Lambda^{(n)}$:
\begin{linenomath*}
\begin{equation} \label{eqn:Lambda}
\mathcal{L}\left(\Lambda^{(n)}\right)=-\frac{1}{T}\sum_{t=1}^{T}\sum_{m=1}^M\Pi^{(n)}_{mt}\log\left(\sum_{k=1}^K\Lambda^{(n)}_{mk}\Gamma^{(n)}_{kt}\right)\rightarrow \min_{\Lambda^{(n)}\in\Omega_\Lambda}.
\end{equation} 
\end{linenomath*}
Here $\Lambda^{(n)}\in\Omega_\Lambda$ is the matrix of conditional probabilities for lead time $n$, which are constrained to be in the feasible set $\Omega_\Lambda=\{\Lambda_{mk}\in[0,1]\forall m,k:\sum_{m=1}^M\Lambda_{mk}=1\forall k\}$. In the figure \ref{fig:summary} schematic these conditional probabilities are shown in panel (e). Note that this is the same cross-entropy loss functional that appears in eSPA, but without the application of Jensen's inequality as here $\Gamma^{(n)}$ is not assumed to be discrete. This nonlinear program can be efficiently solved using an interior-point method such as Ipopt \cite{Wachter2006} (step (v) of figure 1).

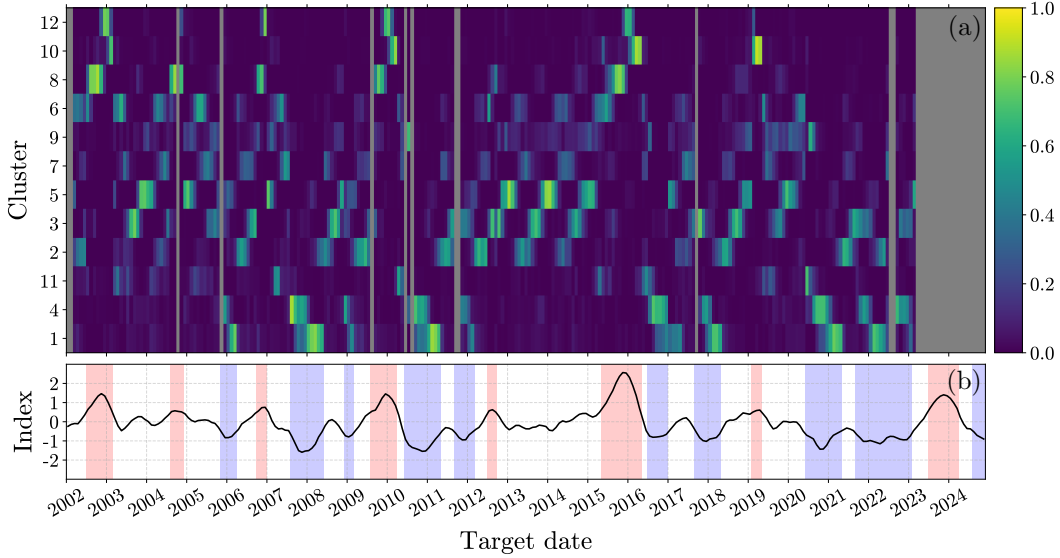
\begin{figure}[t]
\centering
\begin{tikzpicture}
    \node (img1) {\includegraphics[width=\textwidth]{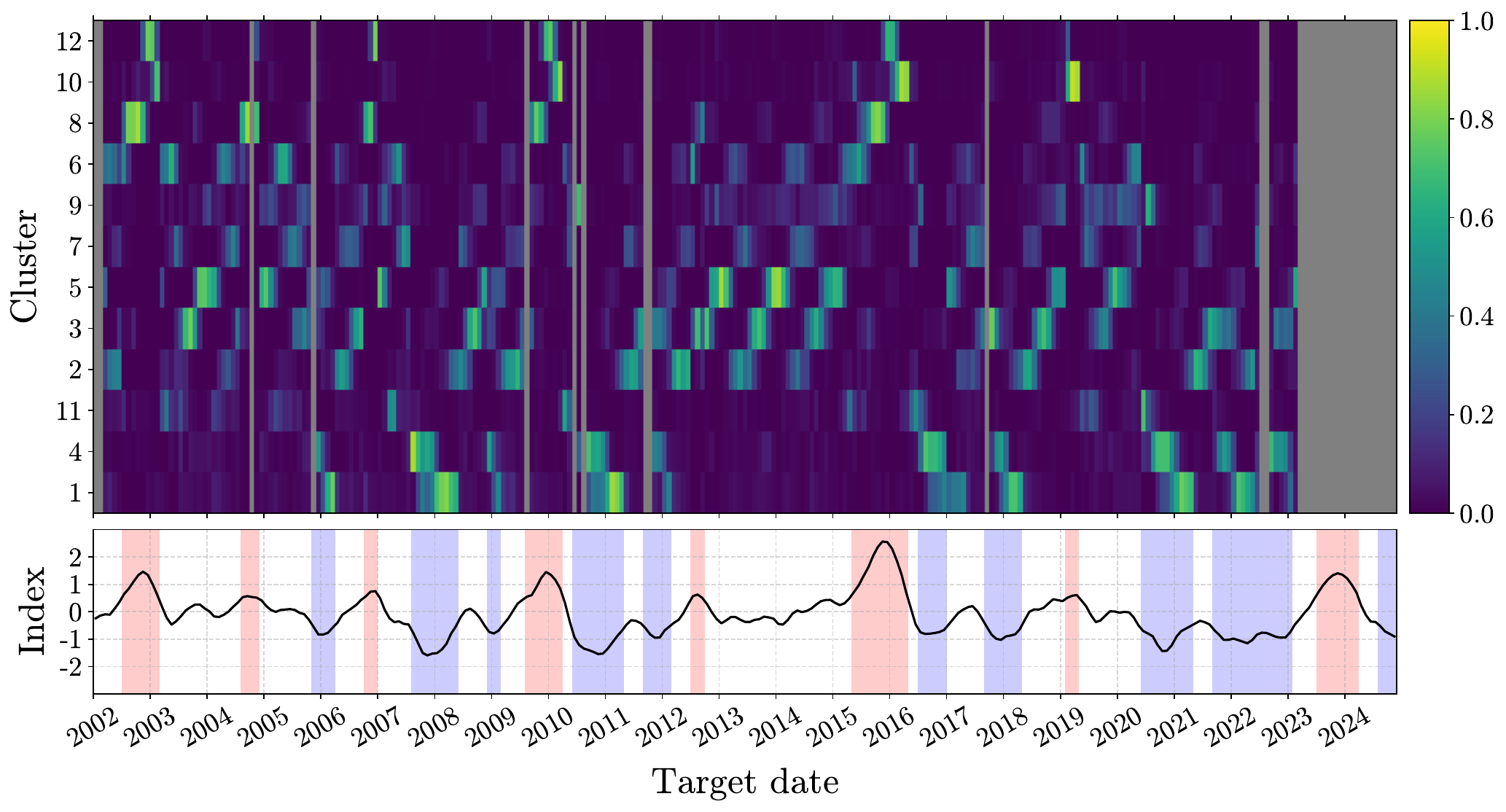}};
    \node[anchor=north east, xshift=-30pt, yshift=-8pt] at (img1.north east) {(a)};
    \coordinate (b) at ($ (img1.north east)!0.64!(img1.south east) $);
    \node[anchor=north east, xshift=-30pt, yshift=0pt] at (b) {(b)};
  \end{tikzpicture}
\caption{The affiliation probabilities $\Gamma^{(n)}$ vs. target date for a lead time of $n$=3 months over the evaluation period of January 2002 to December 2024 (a), and the 3-month running average of the Ni\~no3.4 index over this same period (b). The background shading in (b) corresponds to target dates where the Ni\~no3.4 index is $\ge$0.5$^\circ$ (red) or $\le-$0.5$^\circ$ (blue).}
\label{fig:affiliations}
\end{figure}

Figure \ref{fig:affiliations} shows a plot of $\Gamma^{(n)}$ vs. target date for $n=3$ months lead time. The $y$-axis of figure \ref{fig:affiliations}(a) corresponds to each of the 12 superclusters for this lead time, which are ordered by the expected value of their conditional probabilities $\Lambda^{(n)}_{:,k}$. The grey shading indicates target dates for which there are no correct predictions at this lead time, either because no forecasts were issued for those target dates at this lead time (i.e. at either end of the evaluation period) or because none of the 50 eSPA models in the ensemble correctly predicted the phase of ENSO on those target dates at this lead time. By ordering the $y$-axis by the expected value of each supercluster (i.e. from -1 to 1), clear patterns emerge in the sequence of preferred transitions between superclusters that correlate strongly with the Ni\~no3.4 index. These will be explored explicitly in section \ref{subsec:transitions} below. 

\begin{figure}[t]
\centering
\noindent\includegraphics[width=\textwidth]{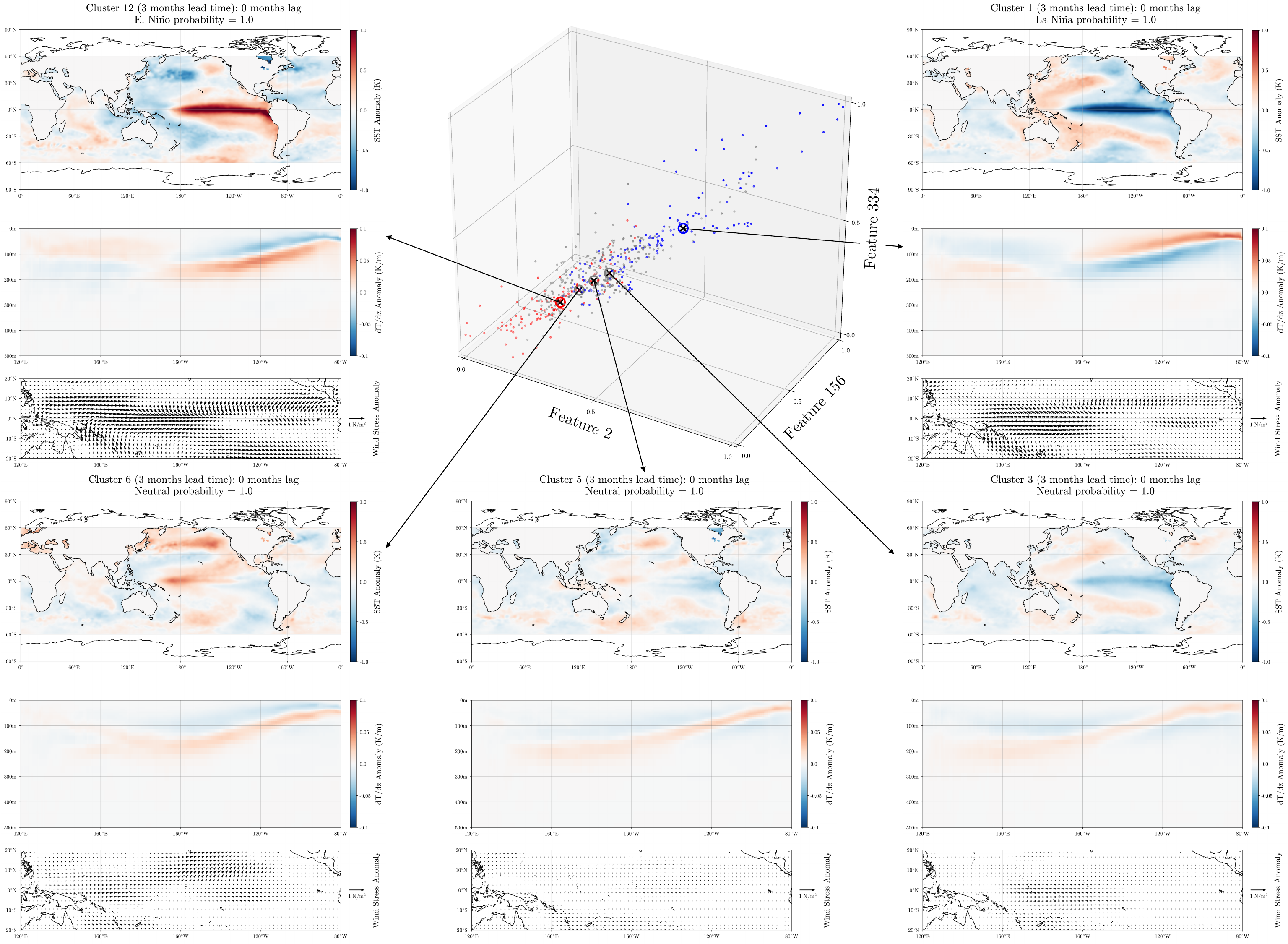}
\caption{Visualisations of the lag-0 composites of SST, $\mathrm dT/\mathrm dz$ and wind stress corresponding to superclusters 12, 6, 5, 3 and 1 at $n$=3 months lead time, as well as the projection of the centroids for these clusters onto the 3 most important feature dimensions at this lead time. The projection of the features $X$ onto these 3 dimensions is also shown, with each point coloured by its corresponding label, i.e. the phase of the Ni\~no3.4 index in 3 months time.}
\label{fig:composites}
\end{figure}

Following the methodology introduced in \citeA{Groom2024}, the centroid $\tilde C^{(n)}_{:,k}$ for each supercluster $k=1,\dots,12$ at lead time $n$ can be used to construct a composite pattern in the original SST, $\mathrm dT/\mathrm dz$ and wind stress fields. This involves first converting to the unweighted centroids $\bar C^{(n)}=\tilde C^{(n)}\odot\left(1/\sqrt{W^{(n)}}\right)$, and then rescaling to the range of the original feature dimensions (i.e. the principal components of each field). At this point the composites for each field $f=1,2,3$ (where $f=1$ is SST, $f=2$ is $\mathrm dT/\mathrm dz$ and $f=3$ is wind stress), can be generated as
\begin{linenomath*}
\begin{equation} \label{eqn:composites}
\mathcal C_k^{(n)}(\mathbf x,l,f)=\sum_{d=1}^{N_f}\bar C^{(n)}_{(d_f+d-1),k}\cdot\mathrm{EOF}_d(\mathbf x,l)
\end{equation} 
\end{linenomath*}
where $\mathrm{EOF}_d(\mathbf x,l)$ is the $d$-th EOF of the $f$-th field (a function of both the spatial coordinates $\mathbf x$ and the lag dimension $l=0,\ldots,11$), $N_f$ is the total number of PCs for field $f$ and $d_f$ is the starting index of the dimensions of field $f$ in the overall list of dimensions $d=1,\ldots,D$. Note: for the $\mathrm dT/\mathrm dz$ there is an extra spatial dimension $z$ in $\mathbf x$ and for the wind stress field there are two components $c\in\{x,y\}$.

Figure \ref{fig:composites} shows a scatter plot of the features $X$ projected onto the 3 most important dimensions at a lead time of $n=3$ months (dimensions 2, 156 and 334 of the vector $W^{(3)}$), with each point $t$ coloured by (the expected value of) its corresponding label, given by $\Pi^{(3)}_{:,t}$, i.e. the phase of the Ni\~no3.4 index 3 months ahead of time $t$. Superimposed on this plot are the centroids for the 3-month lead time superclusters 1, 3, 5, 6 and 12, whose corresponding spatial composites are shown in the surrounding panels. Note that each point in this phase space actually corresponds to a spatiotemporal pattern due to the 12-month embedding, however only the lag-0 spatial patterns are shown in figure \ref{fig:composites}. Also note that although these supercluster centroids were not fit on the features $X$, but rather centroids from eSPA models whose predictions were verified to be correct, an affiliation probability for each instance $X_{:,t}$ ($t=1,\ldots,T$) to each supercluster $\tilde C^{(n)}_{:,k}$ ($k=1,\ldots,12$) can still be calculated by solving equation \ref{eqn:Gamma} using $X$ in place of $\tilde X^{(n)}$. Figure S2 of the supplementary material shows a plot of these alternative affiliation probabilities for 3 months lead time.

\subsubsection{Transition matrices} \label{subsec:transitions}

As mentioned above, the sequences of affiliation probabilities over superclusters exhibit transitions in time that are reflective of preferred pathways between clusters for a given forecast lead time. This observation can be made more explicit by calculating $K\times K$ transition probability matrices $P^{(n)}$ for each lead time. By convention, these are row-stochastic, i.e. $P^{(n)}\in\Omega_P$ where $\Omega_P=\{P_{ij}\in[0,1]\forall i,j:\sum_{j=1}^KP_{ij}=1\forall i\}$. In other words, $P^{(n)}_{ij}$ is the probability of transitioning to cluster $j$ at time $t+1$, given that the system is in cluster $i$ at time $t$. Given that $\Gamma^{(n)}$ is in general fuzzy, $P^{(n)}$ can be solved for using maximum likelihood estimation (MLE), which results in the following (convex) nonlinear program:
\begin{linenomath*}
\begin{equation} \label{eqn:P}
\mathcal{L}\left(P^{(n)}\right)=-\frac{1}{T-1}\sum_{t=1}^{T-1}\sum_{j=1}^K\Gamma^{(n)}_{j,t+1}\log\left(\sum_{i=1}^KP^{(n)}_{ij}\Gamma^{(n)}_{it}\right)\rightarrow \min_{P^{(n)}\in\Omega_P}.
\end{equation} 
\end{linenomath*}
In practice, equation \ref{eqn:P} can be solved using the same interior-point method that is used to solve for $\Lambda^{(n)}$ in equation \ref{eqn:Lambda} by transposing $P^{(n)}$ and converting the row-stochastic constraints $\Omega_P$ to column-stochastic constraints $\Omega_{P^\mathsf{T}}$. Note that equation \ref{eqn:P} can be extended to handle cases where there are gaps in the sequence (e.g. due to missing affiliations on certain target dates) by splitting the sum over $t$ into separate sums over piecewise continuous segments.

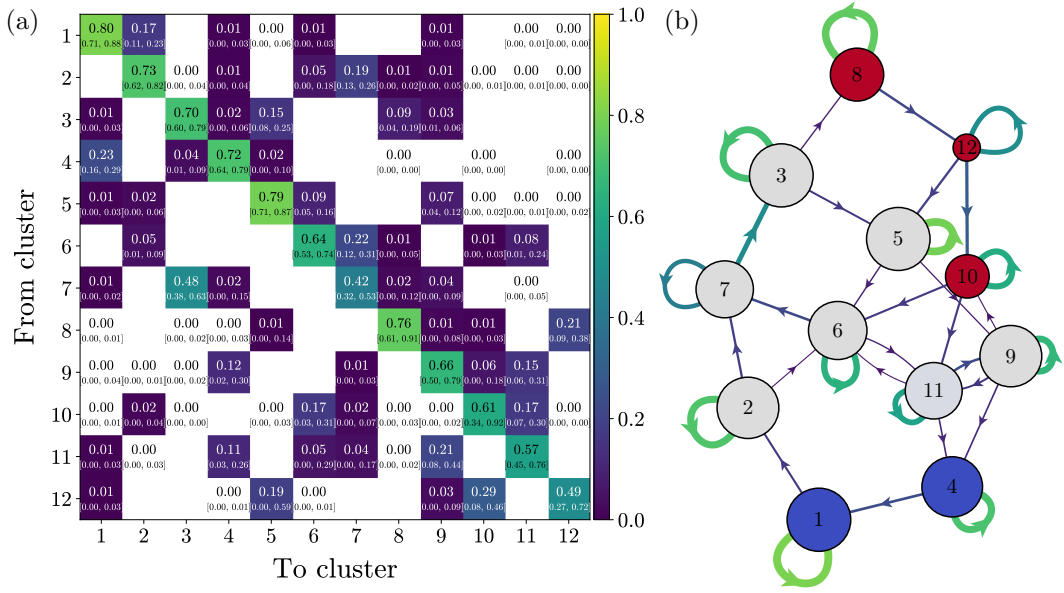
\begin{figure}[t]
\centering
\begin{tikzpicture}
    \node (img1) {\includegraphics[width=\textwidth]{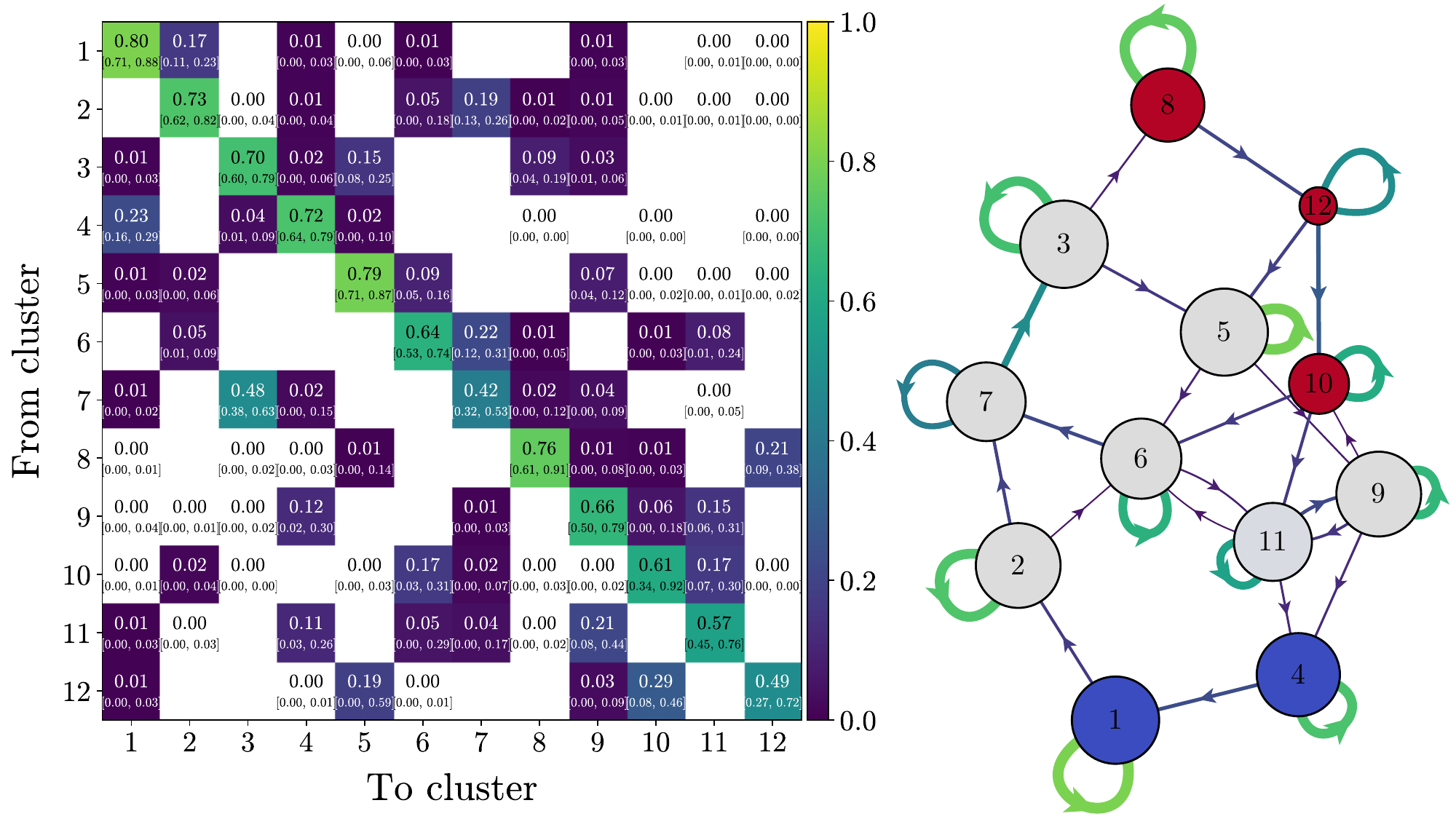}};
    \node[anchor=north west, xshift=0pt, yshift=-4pt] at (img1.north west) {(a)};
    \coordinate (b) at ($ (img1.north west)!0.66!(img1.north east) $);
    \node[anchor=north east, xshift=0pt, yshift=-4pt] at (b) {(b)};
  \end{tikzpicture}
\caption{The transition probability matrix $P^{(n)}$ for a lead time of $n$=3 months, including 95\% confidence interval estimates (a), and a visualisation of the matrix as a directed cyclic graph (b), where only the probabilities $P^{(n)}_{ij}$$\ge$0.05 are shown as edges. The nodes are coloured by the expected value of the conditional probabilities $\Lambda^{(n)}_{:,k}$ for each cluster, and the edges are coloured (and sized) by the transition probabilities $P^{(n)}_{ij}$.}
\label{fig:transitions}
\end{figure}

Figure \ref{fig:transitions}(a) shows a plot of $P^{(n)}$ for $n=3$ months lead time. Each cell $(i,j)$ of the heatmap shows the probability $P^{(3)}_{ij}$, along with 95\% confidence intervals (CIs) on that probability obtained via bias-corrected and accelerated bootstrapping (using the jackknife method to estimate the acceleration parameter). Cells that are white but still contain an annotation are those that have a maximum likelihood estimate of $P^{(3)}_{ij}=0$ but which have a non-zero upper-bound on their 95\% confidence interval, which cells that are entirely white are those with both a zero MLE probability and a zero upper-bound on their 95\% CI. Figure \ref{fig:transitions}(b) shows a visualisation of the matrix in figure \ref{fig:transitions}(a) as a directed cyclic graph (i.e. a Markov chain), where only the probabilities $P^{(3)}_{ij}\ge0.05$ are shown as edges. This visualisation makes it easier to see the dominant pathways between clusters. For example, transitions from La Ni\~na clusters to El Ni\~no clusters must occur via neutral clusters and vice versa.

Formulating the $t\rightarrow t+1$ transitions at each lead $n$ as a (discrete-time) Markov chain also allows for the computation of various important properties of the Markov chain that have a clear physical meaning. Calculating the eigenvalues $\lambda_1,\ldots,\lambda_K$ of $P^{(n)}$, sorted by their magnitude from largest to smallest, allows one to calculate the relaxation time as
\begin{linenomath*}
\begin{equation} \label{eqn:relax}
r=\frac{1}{1-\vert\lambda_2\vert},
\end{equation} 
\end{linenomath*}
which gives a global measure of the average persistence. For $n=3$, $r=6.3$ months, which is very close to the $e$-folding time of the Ni\~no3.4 index of 6.1 months (see figure S3 of the supplementary material), indicating that this set of clusters capture the dynamics of ENSO quite well. In general, the relaxation time $r$ for each lead time $n$ is close to the $e$-folding time, with the average relaxation time across all lead times equal to 6.0 months.

If a Markov chain is irreducible and aperiodic, then (as per the Perron--Frobenius theorem) there is a unique stationary distribution $\pi$ and the $a$-th power of the transition matrix, $P^a$, converges to a rank-one matrix where each row is the stationary distribution, i.e.
\begin{linenomath*}
\begin{equation} \label{eqn:stationary}
\lim_{a\rightarrow\infty}P^a=\mathbf 1\pi,
\end{equation} 
\end{linenomath*}
where $\mathbf 1$ is a column vector with all entries equal to 1. Equivalently, $\pi$ satisfies $P^\mathsf{T}\pi=\pi$ and can therefore be calculated as the leading eigenvector of $P^\mathsf{T}$ (whose eigenvalues are sorted by their magnitude, giving $\lambda_1=1$). From the stationary distribution, one can calculate the mean return time for each state $i$ as $m_i=1/\pi_i$. In figure \ref{fig:transitions}(b), the size of each node has been scaled by the inverse of its respective mean return time, i.e. smaller nodes have a longer mean return time (e.g. $m_{12}=43.9$ months for node 12) whereas larger nodes have a shorter mean return time (e.g. $m_1=8.0$ months for node 1). Other quantities which can be easily calculated from $P$ and/or $\pi$ are the mean exit time from state $i$,
\begin{linenomath*}
\begin{equation} \label{eqn:exit}
e_i=\frac{1}{1-P_{ii}},
\end{equation} 
\end{linenomath*}
and the mean first passage time from state $i$ to state $j$ ($i\ne j)$,
\begin{linenomath*}
\begin{equation} \label{eqn:mfpt}
H_{ij}=\frac{Z_{ii}-Z_{ij}}{\pi_j},
\end{equation} 
\end{linenomath*}
where $Z=\left(I-P+\mathbf1\otimes\pi^\mathsf{T}\right)^{-1}$ is the fundamental matrix of the Markov chain (which exists provided the chain is finite, irreducible and aperiodic). Using the $n=3$ months lead time Markov chain as an example, the mean exit times range from $e_1=5.0$ months to $e_{12}=1.9$ months. Similarly, the minimum and maximum mean first passage times are $H_{7,3}=5.6$ months and $H_{9,12}=86.8$ months respectively. All of these quantities can be calculated/displayed for every Markov chain in the online interactive visualisations accompanying this article.

In addition to calculating $t\rightarrow t+1$ transition probabilities at each lead $n$, it is also possible to calculate transition probabilities $G^{(n)}$ between leads $n$ and $n-1$ for a fixed time $t$, i.e. by solving
\begin{linenomath*}
\begin{equation} \label{eqn:G}
\mathcal{L}\left(G^{(n)}\right)=-\frac{1}{T}\sum_{t=1}^{T}\sum_{j=1}^K\Gamma^{(n-1)}_{j,t}\log\left(\sum_{i=1}^KG^{(n)}_{ij}\Gamma^{(n)}_{it}\right)\rightarrow \min_{G^{(n)}\in\Omega_P}.
\end{equation} 
\end{linenomath*}
This represents the probability of the lead $n-1$ forecast being assigned to cluster $j$ for target time $t$, given that the lead $n$ forecast was assigned to cluster $i$ for target time $t$ (i.e. the forecast for that same target date that was issued the previous month). By concatenating the cluster indices across all lead times into a global list $\{1,\ldots,\mathcal K\}$, each matrix $G^{(n)}$ can be seen as forming one block of a $\mathcal K \times\mathcal K$ upper block-diagonal matrix $\mathcal G$. In other words, $\mathcal G$ is a directed acyclic graph---specifically, a Bayesian network---representing the transitions between cluster states with descending lead time for a given target date. As before with calculating $P^{(n)}$, gaps in the sequence can be handled by restricting the sum over $t$ in equation \ref{eqn:G} to those times $t$ for which both $\Gamma^{(n)}_{:,t}$ and $\Gamma^{(n-1)}_{:,t}$ are defined. It is also important to note that if $\Gamma^{(n-1)}_{:,t}$ in equation \ref{eqn:G} is swapped for $\Pi_{:,t}$ then $G^{(n)}$ is equivalent to (the transpose of) $\Lambda^{(n)}$. Indeed, $\Lambda$ can itself be interpreted as a Bayesian network between two stochastic processes $X(t)$, discretised into states $S^X=\{S^X_1,\ldots,S^X_K\}$, and $Y(t)$, discretised into states $S^Y=\{S^Y_1,\ldots,S^Y_M\}$, with corresponding probabilistic representations $\Gamma^X(t)=\Gamma(t)$ and $\Gamma^Y(t)=\Pi(t)$ \cite{Horenko2020}. Therefore, for completeness, $\Lambda^{(1)}$ can also be added to $\mathcal G$, making it a $(\mathcal K+M) \times(\mathcal K+M)$ matrix.

\begin{figure}[H]
\centering
\begin{tikzpicture}
    \node (img1) {\includegraphics[height=0.98\textheight]{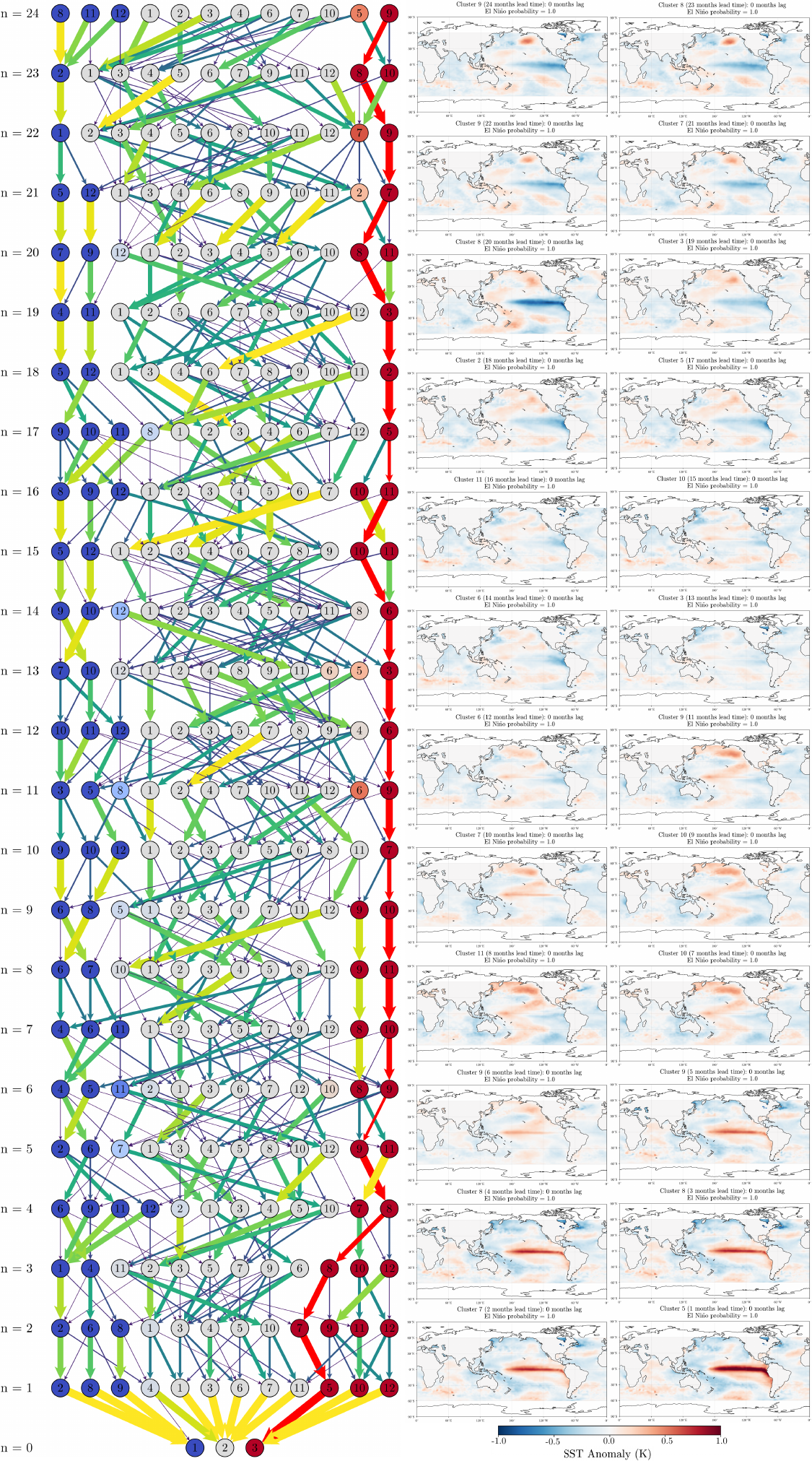}};
    \node[anchor=south west, xshift=28pt, yshift=2pt] at (img1.south west) {(a)};
    \node[anchor=south west, xshift=186pt, yshift=2pt] at (img1.south west) {(b)};
  \end{tikzpicture}
  \vspace{-1em}
\caption{(a) The Bayesian network $\mathcal G$ of lead $n$$\rightarrow$$n$$-1$ transition probabilities. (b) SST composites for each cluster on the most probable path between node 9 ($n$$=$24 months) to node 3 ($n$$=$0 months), highlighted in red.}
\label{fig:vertical}
\end{figure}

The transition probabilities represented by the Bayesian network $\mathcal G$ therefore give a measure of the consistency of the superclusters across different lead times, where in the ideal case the conditional probabilities calculated by traversing the graph from node $k\in\{1,\ldots,K^{(n)}\}$ at lead $n$ to node $m\in\{1,\ldots,M\}$ at lead 0 (i.e. class $m$) would be identical to $\Lambda^{(n)}_{mk}$, i.e. the conditional probability that is calculated directly between cluster $k$ at lead $n$ and class $m$. In addition, $\mathcal G$ can be used to find the most probable path between node $i$ at lead $n$ and node $j$ at lead $l<n$ by computing the Dijkstra path between these two nodes, with the cost of edge each set to $-\log(p)$ where $p$ is the edge probability. Figure \ref{fig:vertical} shows a plot of $\mathcal G$, with the most probable path from node 9 at $n=24$ months to node 3 at $n=0$ months (i.e. the El Ni\~no class) highlighted in red. Alongside this plot are the SST composites for each cluster on the most probable path for $n=24,23,\ldots,1$ months. This sequence of composites represents the canonical 24-month pathway to an El Ni\~no event that is captured by the set of superclusters. A similar canonical pathway to a La Ni\~na event can also be generated, for example by computing the Dijkstra path between node 8 at $n=24$ months to node 1 at $n=0$ months (shown in figure S4 of the supplementary material). For any particular target date in the evaluation period, the specific pathway taken to the event at time $t$ is determined by the affiliations $\Gamma_{:,t}^{(n)}$ at each lead $n$, which are fuzzy in general---a reflection of the diversity among events. A similar sequence of composites can be therefore generated for each particular event by calculating the reconstructed feature vector at each lead for time $t$ as
\begin{linenomath*}
\begin{equation} \label{eqn:Xhat}
\hat X^{(n)}_{:,t}=\tilde C^{(n)}\cdot\Gamma_{:,t}^{(n)}
\end{equation} 
\end{linenomath*}
and then constructing the composite patterns in the original SST, $\mathrm dT/\mathrm dz$ and wind stress fields using equation \ref{eqn:composites}. This enables case studies to be conducted of the pathways for particular events of interest in the evaluation period, which will be explored in detail in section \ref{subsec:case-studies}.

\subsubsection{Imputing missing affiliations} \label{subsec:imputing}
Aside from providing insight into the dynamics that are accurately captured by the distilled models, both the $t\rightarrow t+1$ and $n\rightarrow n-1$ transition matrices can also be used to impute missing affiliations (i.e. on target dates for which there are no correct predictions) at each lead time. The process for doing this imputation is as follows. First, gaps in the affiliation sequence for a given lead time $n$ are found and the affiliations either side of the gap are stored. For simplicity we will just consider the case of a 1-month gap for now, in which case the stored affiliations are $\Gamma_{:,t-1}^{(n)}$ and $\Gamma_{:,t+1}^{(n)}$. Next, both the $t\rightarrow t+1$ transition probability matrix $P^{(n)}$ and its reverse-time equivalent $P_r^{(n)}$ (calculated by solving equation \ref{eqn:P} with $\Gamma_{j,t-1}^{(n)}$ instead of $\Gamma_{j,t+1}^{(n)}$) are used to evolve the affiliations $\Gamma_{:,t-1}^{(n)}$ and $\Gamma_{:,t+1}^{(n)}$ to time $t$, e.g. 
\begin{linenomath*}
\begin{equation} \label{eqn:forward}
\Gamma_{:,t}^{(f)}=P^\mathsf{T}\Gamma_{:,t-1}
\end{equation} 
\end{linenomath*}
is the forward evolution and 
\begin{linenomath*}
\begin{equation} \label{eqn:backward}
\Gamma_{:,t}^{(r)}=P_r^\mathsf{T}\Gamma_{:,t+1}
\end{equation} 
\end{linenomath*}
is the backward evolution. The evolved affiliations are then combined with $\Lambda^{(n)}$ to give the predicted class probability vector for time $t$, i.e.
\begin{linenomath*}
\begin{equation} \label{eqn:probs}
\hat\Pi_{:,t}=\Lambda\Gamma_{:,t}
\end{equation} 
\end{linenomath*}
is calculated for both $\Gamma_{:,t}^{(f)}$ and $\Gamma_{:,t}^{(r)}$. If one or both of the predicted class probability vectors at time $t$ correctly predicts the true class with majority probability then the affiliation is saved (in the case where both forward and backward predictions are correct the one with the higher majority probability is saved), otherwise the gap in the sequence is kept. For gaps of longer than 1 month this process is repeated iteratively.

Once this ``horizontal'' imputation of affiliations at each lead $n$ has been performed, a second, ``vertical'' imputation is performed using the $n\rightarrow n-1$ transition matrices $G^{(n)}$ and their equivalent $n\rightarrow n+1$ transition matrices $G_r^{(n)}$ (i.e. calculated by solving equation \ref{eqn:G} with $\Gamma_{j,t}^{(n+1)}$ instead of $\Gamma_{j,t}^{(n-1)}$). Specifically, for each time $t$ we look for any missing affiliations in the sequence $\{\Gamma_{:,t}^{(1)},\ldots,\Gamma_{:,t}^{(24)}\}$ and, as before, store the affiliations either side of any gaps that are found, i.e. $\Gamma_{:,t}^{(n+1)}$ and $\Gamma_{:,t}^{(n-1)}$ for simplicity. Next, $G^{(n+1)}$ and $G_r^{(n-1)}$ are used to evolve $\Gamma_{:,t}^{(n+1)}$ and $\Gamma_{:,t}^{(n-1)}$ to lead $n$ in the same manner as equations \ref{eqn:forward} and \ref{eqn:backward}. The evolved affiliations are then combined with $\Lambda^{(n)}$ using equation \ref{eqn:probs} to give the predicted class probability vector $\hat\Pi^{(n)}_{:,t}$, with the same criteria as before used to determine whether either of the evolved affiliations should be saved. Again, for gaps of longer than 1 month this process is repeated iteratively. The final result of this imputation process for $n=3$ months lead time is shown in figure \ref{fig:affiliations-imputed}. When compared with figure \ref{fig:affiliations}, there are now only 2 gaps remaining in the sequence vs. 8 gaps in the original sequence. 

\begin{figure}[t]
\centering
\includegraphics[width=\textwidth]{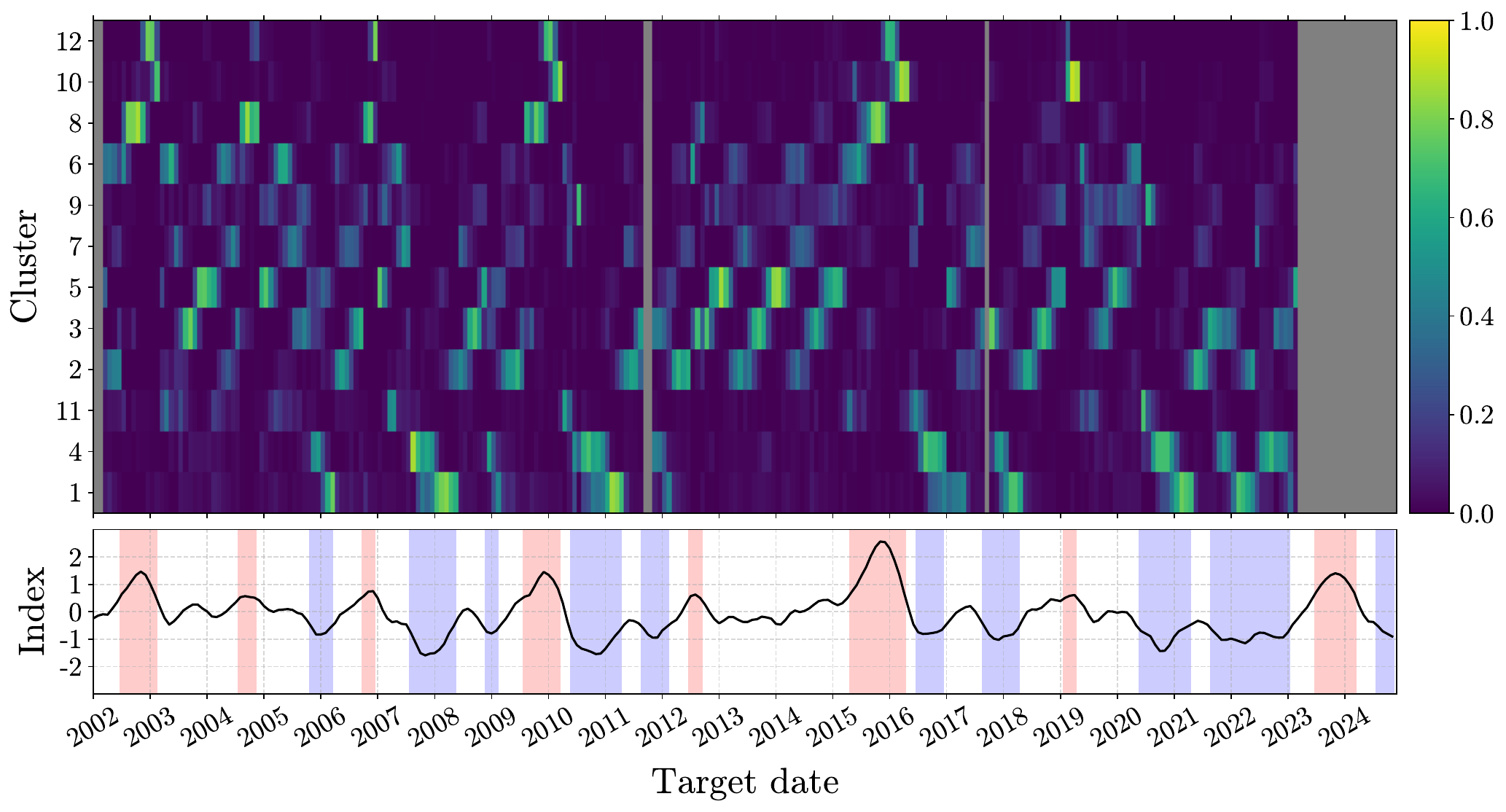}
\caption{The imputed affiliation probabilities $\Gamma^{(n)}$ vs. target date for a lead time of $n$=3 months. See the figure \ref{fig:affiliations} caption for further details.}
\label{fig:affiliations-imputed}
\end{figure}

\subsection{Composites of diagnostic fields}
Once an affiliation sequence $\Gamma^{(n)}$ has been calculated (and any missing values imputed), it is possible to use this sequence to generate composites for any gridded data that is available over the same range of target dates that $\Gamma^{(n)}$ is defined for. This is done in a two-step procedure. First, composites for each cluster $k$ at the desired lead time $n$ are generated by averaging the gridded data $F(\mathbf x,t)$, weighted by each cluster's affiliation probabilities, i.e.
\begin{linenomath*}
\begin{equation} \label{eqn:future-composites1}
\mathcal C_k^{(n)}(\mathbf x)=\frac{\sum_{t=1}^T\Gamma^{(n)}_{kt}\odot F(\mathbf x,t)}{\sum_{t=1}^T\Gamma^{(n)}_{kt}}.
\end{equation} 
\end{linenomath*}
Then, for each target date $t$, these cluster composites $\mathcal C_k^{(n)}$ are blended according to that target date's affiliation probability distribution:
\begin{linenomath*}
\begin{equation} \label{eqn:future-composites2}
\mathcal C^{(n)}(\mathbf x, t)=\frac{\sum_{k=1}^K\Gamma^{(n)}_{kt}\odot \mathcal C_k^{(n)}(\mathbf x)}{\sum_{k=1}^K\Gamma^{(n)}_{kt}}.
\end{equation} 
\end{linenomath*}

This procedure allows for the visualisation of not only the precursors of particular events of interest, but also the future state that is being predicted for each target date at a given lead. Figure \ref{fig:future-composite} shows an example of this, where the diagnostic field is (detrended) SST on the original $0.25^\circ\times0.25^\circ$ grid for a lead time of $n=24$ months ahead. Figure \ref{fig:future-composite}(a) shows the individual cluster composites of this field that are generated from equation \ref{eqn:future-composites1}, while figure \ref{fig:future-composite}(b) shows the combined composite generated from equation \ref{eqn:future-composites2} for a target date of January 2016. Compared to the (detrended) ground truth SST field for January 2016 shown in figure \ref{fig:future-composite}(c), the patterns are broadly quite similar (the area-weighted pattern correlation is 0.75), although the amplitude of the anomalies in the composite is reduced relative to the ground truth. It is important to stress that, although this data was used to derive the training inputs for the eSPA models, here it is being treated in a diagnostic rather than prognostic manner. Similar composites are able to be generated over the same range of dates for fields that are not related to the training inputs in any way, for example subsurface temperature as shown in figure S5 of the supplementary material.

\begin{figure}[t]
\centering
\begin{tikzpicture}
    \node (img1) {\includegraphics[width=\textwidth]{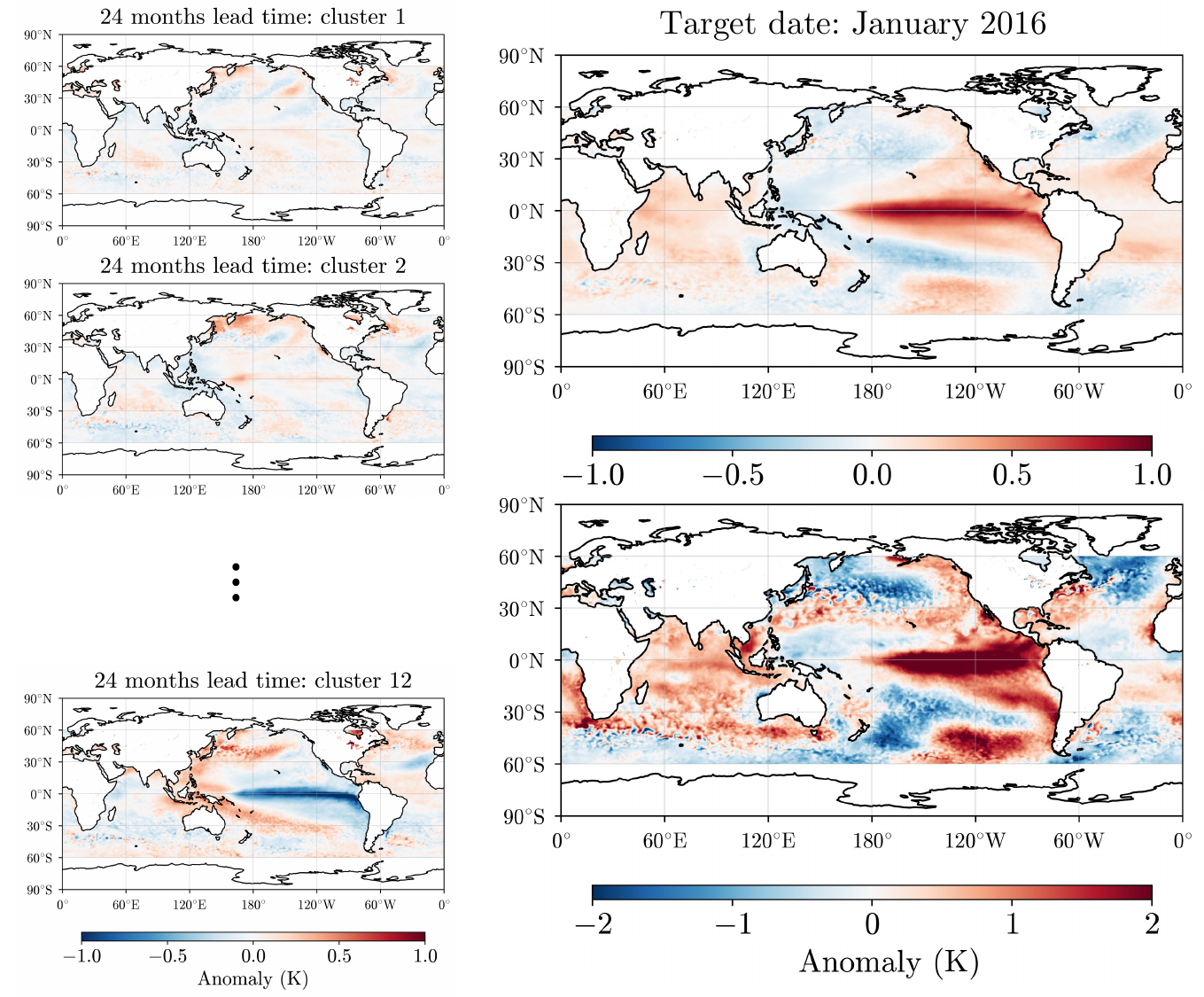}};
    \node[anchor=north west, xshift=6pt, yshift=5pt] at (img1.north west) {(a)};
    \coordinate (b) at ($ (img1.north west)!0.46!(img1.north east) $);
    \node[anchor=north east, xshift=0pt, yshift=5pt] at (b) {(b)};
    \coordinate (c) at ($ (img1.south west)!0.46!(img1.south east) $);
    \node[anchor=north east, xshift=0pt, yshift=192pt] at (c) {(c)};
  \end{tikzpicture}
\caption{(a) Diagnostic composites of the $n$-month ahead SST field for clusters $1,\ldots,12$ at $n=24$ months. (b) The combined diagnostic composite of 24-month ahead SST for a target date of January 2016. (c) The detrended ground truth SST field for January 2016.}
\label{fig:future-composite}
\end{figure}
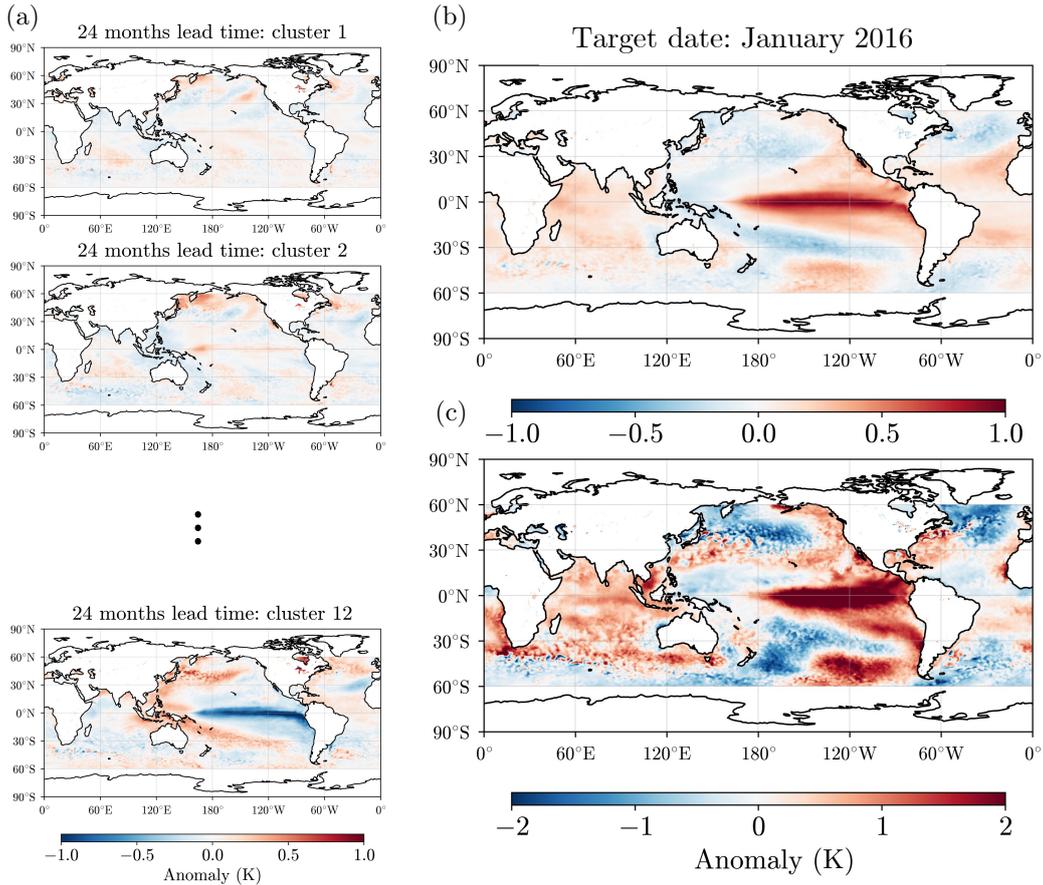

\subsection{Skill metrics}
The following metrics are used both for scoring individual eSPA models as well as for assessing the ensemble predictions against ground truth data. The ranked probability score (RPS) is defined as 
\begin{linenomath*}
 \begin{equation} \label{eqn:RPS}
    \mathrm{RPS} = \frac{1}{T}\sum_{t=1}^{T}\sum_{m=1}^{M}\left(\sum_{j=1}^{m}\hat\Pi_{j,t}-\sum_{j=1}^{m}\Pi_{j,t}\right)^2,
 \end{equation}
\end{linenomath*}
where $\hat\Pi_{m,t}$ and $\Pi_{m,t}$ are the predicted and true probabilities for class $m=1,\ldots,M$ and instance $t=1,\ldots,T$ respectively. The RPS thus penalises predictions that are further away from the ground truth more heavily in cases where the classes are ordinal, with a worst-case value of $M-1$. Similarly, the ranked probability skill score (RPSS) is defined as
\begin{linenomath*}
 \begin{equation} \label{eqn:RPSS}
    \mathrm{RPSS} = 1-\frac{\mathrm{RPS}}{\mathrm{RPS}_c}
 \end{equation}
\end{linenomath*}
where $\mathrm{RPS}_c$ denotes the RPS that is obtained when using climatological probabilities for the predictions \cite{Weigel2007}. By definition, a positive RPSS denotes skill relative to climatology, with a value of 1 denoting perfect skill. 

\section{Results} \label{sec:results}

\subsection{Skill assessment} \label{subsec:skill}

To demonstrate that the distillation procedure does not result in a decrease in skill relative to that of the original ensemble of eSPA models, an assessment of the probabilistic skill is performed. Specifically, we compare the ranked probability skill score of predictions $\hat \Pi$, calculated using equation \ref{eqn:probs}, based on 
\begin{enumerate}
    \item the (imputed) affiliation probabilities of centroids $\tilde{X}^{(n)}$ from correct eSPA models,
    \item affiliation probabilities of the most recent feature vector $X_{:,t}$ at each time $t$,
    \item the mean predictions from the original ensemble of eSPA models,
\end{enumerate}
for each target date and lead time in the evaluation period. Note that option 1 relies on knowledge of which eSPA models in the ensemble are correct vs. incorrect and thus is not representative of real-time forecasting skill. It is merely included to indicate the skill that could be obtained if the ensemble of eSPA models was able to be perfectly filtered (for example using an approach similar to the meta-model proposed in \citeA{Groom2025}) and all incorrect models were removed (note: dates for which there are gaps in the affiliation sequence, e.g. due to no correct models, are also not included in the skill assessment for these centroid affiliation probabilities). Option 2 does not rely on such knowledge and therefore represents an estimate of real-time forecasting skill (noting the caveats given above in section \ref{subsec:forecast-procedure} regarding the use of all available data to generate the PCs) since it only relies on relating each new observation $X_{:,t}$ to the existing set of superclusters $\tilde{C}^{(n)}$ to generate a prediction for each lead time $n$ at time $t$. By construction, there are no gaps in these alternate affiliation sequences over observations and therefore this skill assessment includes all dates in the evaluation period. 

\begin{figure}[t]
\centering
\includegraphics[width=\textwidth]{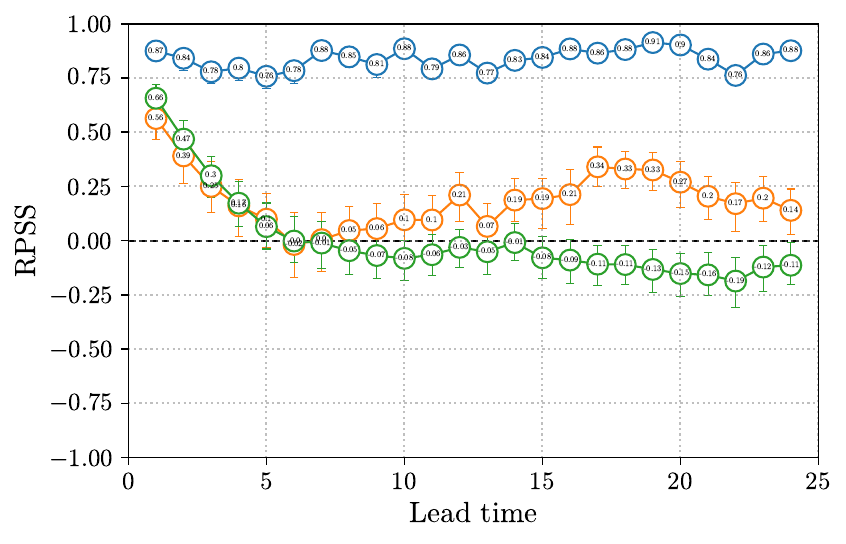}
\caption{Ranked probability skill score vs. lead time, calculated for predictions based on the affiliation probabilities of centroids $\tilde{X}^{(n)}$ from correct eSPA models (blue), affiliation probabilities of features $X$ (orange) as well as the mean predictions from the original ensemble of eSPA models (green). Error bars correspond to 95\% confidence intervals, calculated using bootstrapping.}
\label{fig:skill}
\end{figure}

Figure \ref{fig:skill} shows the RPSS as a function of lead time $n$ for each of these three approaches. The RPSS of the centroid affiliation probabilities, shown in blue, is greater or equal to 0.76 across all lead times (note: this calculation excludes the subset of target dates at each lead time for which there are still missing affiliations after imputation). This shows that the various quantities that are derived from these affiliation sequences are based on predictions that are both very accurate and very confident. The RPSS of the original ensemble of eSPA models, shown in green, dips below 0 for lead times greater than $n=6$ months, and is significantly less than 0 (at the 95\% confidence level) for lead times greater than $n=16$ months. By comparison, the RPSS of the feature affiliation probabilities, shown in orange, remains above 0 for all lead times except $n=6$ months. Furthermore, it is significantly greater than 0 for all lead times except for $5\le n \le 11$ and $n=13$. The slight reduction in skill relative to the original eSPA ensemble at short lead times is attributed to the exclusion of categorical feature dimensions when calculating the assignments in equation \ref{eqn:Gamma}, which (as shown in section \ref{subsec:importance} below) are most important at lead times of $n\le 6$ months due to the persistence of ENSO itself being a relevant predictor at these lead times. While in principle these categorical feature dimensions could be included in equation \ref{eqn:Gamma}, this would result in having to solve a more cumbersome nonlinear program rather than the quadratic program that is solved in the current formulation. Thus, as predictive skill is only an ancillary focus in this paper, that extension is deferred to future work. Similarly, it is possible to make further improvements to the RPSS of the feature affiliation probabilities by including more superclusters than the default choice of $K=12$. However, here we adopt a principle of maximal parsimony and avoid overcomplicating the distillation procedure above what is strictly necessary for our primary aim of interpretability.

\subsection{Feature importance} \label{subsec:importance}
In addition to aggregating clusters from the correct models in the ensemble into superclusters, a similar procedure can be performed for the feature importance vector $W$. Specifically, for each lead time $n$ we construct an aggregated feature importance vector $W^{(n)}$ that is the average of the $W$ vectors of all the correct eSPA models in the ensemble making forecasts at that lead time. Further stratification by target season or class (or both) is also possible. Figure \ref{fig:importance} shows plots of $W^{(n)}$ for lead times of 3, 6, 9, 12, 18 and 24 months, i.e. equivalent plots to those of figure 6 in \citeA{Groom2024} but for an ensemble average rather than individual models (and for features derived using SSA rather than PCA).

\begin{figure}[t] 
\centering
\includegraphics[width=\textwidth]{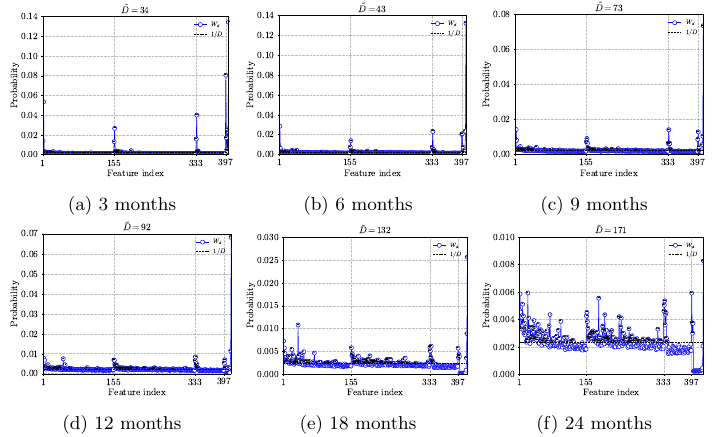}
\caption{The (aggregated) feature importance vector $W^{(n)}$ at various lead times. Features 1-154 are the SSA modes of SST (in order of explained variance), features 155-332 are the SSA modes of $\mathrm dT/\mathrm dz$, features 333-396 are the SSA modes of wind stress, while features $\ge$397 are the Ni\~no3.4 index plus the categorical features containing predictions from earlier lead times.}
\label{fig:importance}
\end{figure}

$W^{(n)}$ exhibits similar trends with lead time to those observed in \citeA{Groom2024}, whereby more importance is placed on the higher-order SSA modes in each field as $n$ increases, indicating that extra-tropical and inter-basin patterns become increasingly important for making successful predictions at longer lead times. The title of each subplot in figure \ref{fig:importance} gives the number of features in $W^{(n)}$ that are greater than the maximum entropy limit of $1/D$ (where $D$ is the total number of features at lead time $n$), i.e.
\begin{linenomath*}
 \begin{equation} \label{eqn:Dtilde}
    \tilde D=\sum_{d=1}^D\mathbbm{1}\left\{W_d>\frac{1}{D}\right\}
 \end{equation}
\end{linenomath*}
where $\mathbbm{1}$ is the indicator function, which is one measure of the ``effective dimension'' of the feature space that is important for predicting ENSO phase at that lead time. This can be used to trim the number of features in $X$ by removing those SSA modes in the SST, $\mathrm dT/\mathrm dz$ and wind stress fields that are never important at any lead time (i.e. $W_d^{(n)}\le1/D\,\forall\, n=1,\ldots,24$). Doing so results in a smaller overall number of PCs being retained: 154, 178 and 64 in each field respectively, which we subsequently denote as $N_1$, $N_2$ and $N_3$. This provides a data-driven method for determining the appropriate number of PCs that are relevant for the prediction task, rather than relying on subjective estimates based on explained variance. Figures S6 and S7 of the supplementary material show the full set of plots of $W^{(n)}$ across all lead times.

\begin{figure}[t]
\centering
\begin{tikzpicture}
    \node (img1) {\includegraphics[width=0.85\textwidth]{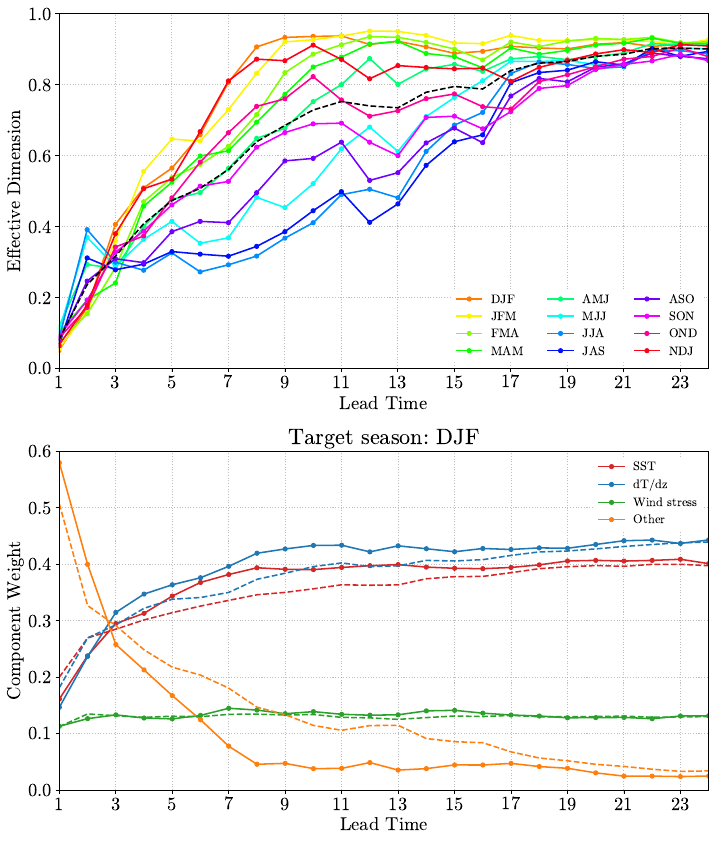}};
    \node[anchor=south east, xshift=20pt, yshift=-10pt] at (img1.north west) {(a)};
    \coordinate (b) at ($ (img1.north west)!0.5!(img1.south west) $);
    \node[anchor=south east, xshift=20pt, yshift=-10pt] at (b) {(b)};
  \end{tikzpicture}
\caption{Normalised effective dimension (a) and component weights for target season DJF (b) vs. lead time. Dashed lines in both plots indicate the all-seasonal average.}
\label{fig:effective-dimension}
\end{figure}
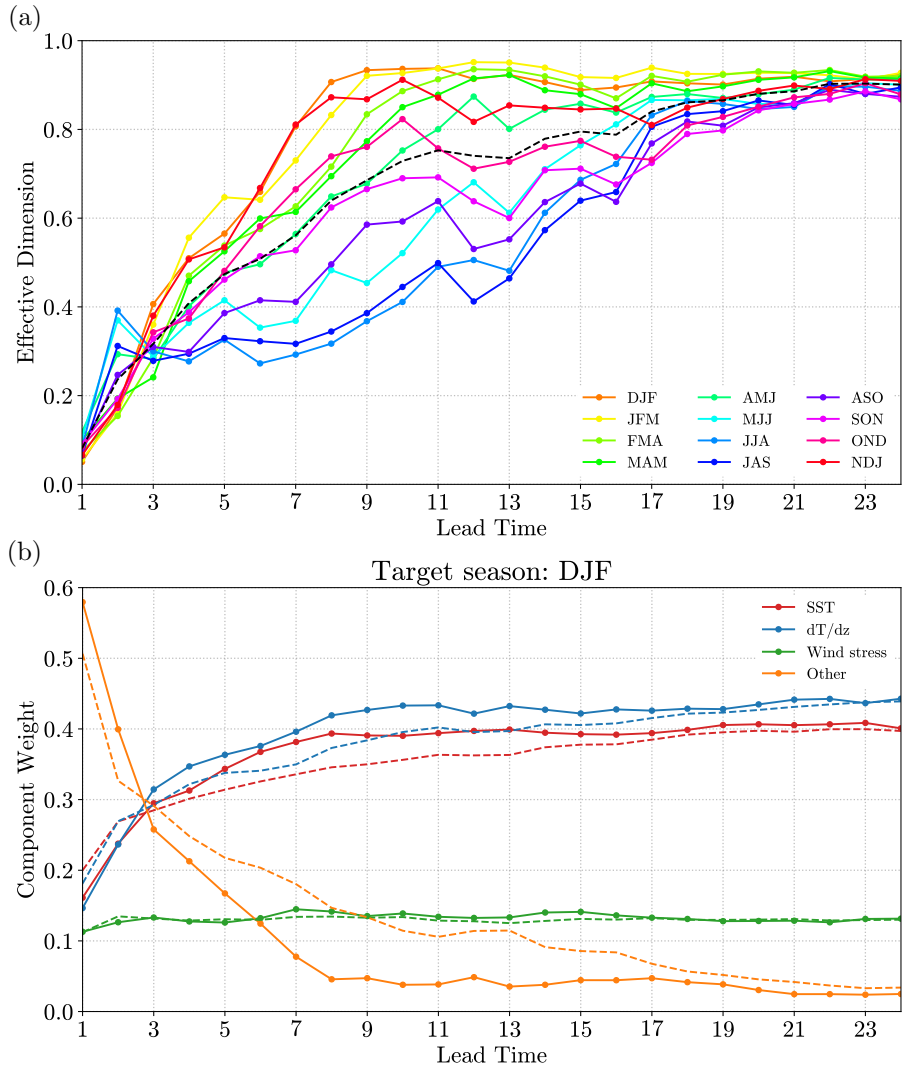

While the expression $W_d>1/D$ is suitable for determining which features are important, in expectation $\tilde D\rightarrow D/2$ as the hyperparameter $\varepsilon_E\rightarrow \infty$ (see \ref{app:effective-dimension} for a proof). For a better measure of the effective dimension of $W^{(n)}$, and thus the sparsity of the feature space at a given lead time, we use the following expression
\begin{linenomath*}
 \begin{equation} \label{eqn:Deff}
    D_{\mathrm{eff}} = \exp\left(-\sum_{d=1}^DW_d\log\left(W_d\right)\right),
 \end{equation}
\end{linenomath*}
i.e. the exponential of the entropy of $W$ \cite{Hill1973}, which has the required property that $D_\mathrm{eff}\rightarrow 1$ as $\varepsilon_E\rightarrow 0$ and $D_\mathrm{eff}\rightarrow D$ as $\varepsilon_E\rightarrow \infty$. Normalising by $D$ therefore gives a quantity that varies between $1/D$ and $1$, depending on the sparsity of the probability distribution. This quantity is plotted in figure \ref{fig:effective-dimension}a as a function of lead time, stratified by target season. From this plot, some important trends can be observed. At short lead times ($\le3$ months) there is less variation in the seasonally-stratified feature importance about the all-seasonal average, given by the dashed line in the plot, and the effective dimension is low, indicating that only the leading SSA modes in each field (as well as the index itself) are important. Similar behaviour is observed at long lead times ($\ge18$ months), although now the effective dimension is high, indicating that almost all SSA modes are important to some degree and therefore the complexity of the input space (features + labels) is higher relative to short lead times. However, for the intermediate lead times the variability of the feature importance with target season is much greater: the effective dimension is typically lowest for target seasons MJJ, JJA and JAS and highest for NDJ, DJF and JFM. This can be explained by considering the Boreal spring predictability barrier. For example, predictions at these lead times targetting Boreal winter need to successfully predict the onset of the next event, through the spring predictability barrier. This requires comparatively more information and therefore the effective dimension (i.e. the entropy of the probability distribution $W$) is higher. Similarly, predictions at these lead times targetting Boreal summer only need to successfully predict the decay of the current event---a task which is comparatively easier---therefore requiring less information on average and a lower effective dimension / entropy of the probability distribution $W$. 

Additional insight into the relative importance of each field may be gained by inspecting plots of the type shown in figure \ref{fig:effective-dimension}b, which shows the total importance (calculated as $\sum_{d=1}^{N_f}W^{(n)}_{d_f+d-1}$ for $f=1,2,3$) assigned to each field as a function of lead time, in this case restricted to target season DJF (with dashed lines indicating the all-seasonal averages). The remaining probability mass in $W^{(n)}$, belonging to the index and categorical features, is summed and also plotted as ``other''. For lead times $n>3$ months, the most important contributions come from the SST and $\mathrm dT/\mathrm dz$ fields as expected, with $\mathrm dT/\mathrm dz$ being the most important and SST close behind. Across all lead times, the importance of the wind stress field is relatively constant, meaning that the steady increase with lead time in importance assigned to $\mathrm dT/\mathrm dz$ and SST is mostly from a steady decrease in importance assigned to the index and categorical features, i.e. persistence is a less important predictor at longer lead times. For the particular target season of DJF shown in figure \ref{fig:effective-dimension}b, at the intermediate lead times of $3 < n < 18$ months the importance assigned to the index and categorical features is lower than the all-seasonal average, while the importance assigned to the $\mathrm dT/\mathrm dz$ and SST fields is higher than the all-seasonal average. This is for the same reason as the trend in the effective dimension at these lead times: predicting the onset of the next event through the Boreal spring barrier requires placing less emphasis on the persistence of the current event and more emphasis on the (potentially) weaker signals contained in the $\mathrm dT/\mathrm dz$ and SST fields, which are likely to be spread out over a wide range of SSA modes in those fields. Similar conclusions can be drawn when inspecting the same plot of component weights for a target season of MJJ (figure S8 in the supplementary material), where now the importance assigned to the index and categorical features over this range of lead times is higher than the all-seasonal average, while the importance assigned to the $\mathrm dT/\mathrm dz$ and SST fields is lower. 

Finally, we conclude this section by demonstrating how the feature importance plots in figure \ref{fig:importance} can be converted to spatial importance maps that highlight the geographical regions that receive the most weight when making predictions at a given lead time. To convert the aggregated feature-importance vectors $W^{(n)}\in\mathbb{R}^{D}$ into spatial importance maps in physical space, the SSA bases of each field (i.e. the EOFs) are combined with the weights $W^{(n)}_j$ assigned to each PC from that field to form a normalised importance attribution map at each lead time. The index/categorical features are excluded prior to this step, with the truncated $W^{(n)}$ renormalised so that $\sum W^{(n)}=1$. The procedure for generating the spatial importance maps can be summarised as follows:
\begin{enumerate}
    \item Load the lagged SSA EOFs (embedding length $L=12$, i.e.\ lags $l=0,\ldots,11$) for the following fields:
    \begin{enumerate}
        \item $\mathrm{EOF}^{\mathrm{SST}}_{j}(\mathbf{x},l)$ for $j=1,\ldots,N_1$ (denoted as field $f=1$),
        \item $\mathrm{EOF}^{\mathrm dT/\mathrm dz}_{j}(\mathbf{x},z,l)$ for $j=1,\ldots,N_2$ (denoted as field $f=2$),
        \item $\mathrm{EOF}^{\tau}_{j,c}(\mathbf{x},l)$ for $j=1,\ldots,N_3$ and $c\in\{x,y\}$ (denoted as field $f=3$),
    \end{enumerate}
    as well as per-gridpoint standard deviations $\sigma^{(f)}(\mathbf x)$ of each field and land-sea masks $\Omega_f$. 
    \item For each lead time $n$, split $W^{(n)}$ into field-specific blocks $W=(W^{(1)},W^{(2)},W^{(3)})$. These blocks define the relative importance of each SSA mode within each field.
    \item For a given lag $l$, compute per-field maps proportional to the squared, standardised EOF amplitudes, weighted by both $W^{(f)}$ and the eigenvalue spectrum of the SSA decomposition. Specifically, given the eigenvalues $\lambda^{(f)}_{j}$ for fields $f\in\{1,2,3\}$, first compute the following (unnormalised) maps: 
    \begin{linenomath*}
    \begin{align}
    \tilde I^{(1)}(\mathbf{x};l) &= \sum_{j=1}^{N_1} W^{(1)}_{j}\lambda^{(1)}_{j}
    \left(\frac{\mathrm{EOF}^{(1)}_{j}(\mathbf{x},l)}{\sigma^{(1)}(\mathbf{x})}\right)^{2}, \\
    \tilde I^{(2)}(\mathbf{x},z;l) &= \sum_{j=1}^{N_2} W^{(2)}_{j}\lambda^{(2)}_{j}
    \left(\frac{\mathrm{EOF}^{(2)}_{j}(\mathbf{x},z,l)}{\sigma^{(2)}(\mathbf{x},z)}\right)^{2}, \\
    \tilde I^{(3)}(\mathbf{x};l) &= \sum_{j=1}^{N_3} W^{(3)}_{j}\lambda^{(3)}_{j}
    \sum_{c}\left(\frac{\mathrm{EOF}^{(3)}_{j,c}(\mathbf{x},l)}{\sigma^{(3)}(\mathbf{x})}\right)^{2}.
    \end{align}
    \end{linenomath*}
    Each field is then normalised by its own total ($W$-weighted) variance, i.e.
    \begin{linenomath*}
    \begin{equation}
    I^{(f)} \;=\; \frac{\tilde I^{(f)}}{\sum_{j}W^{(f)}_{j}\lambda^{(f)}_{j}}.
    \end{equation}
    \end{linenomath*}
    This weighting by $W^{(f)}_{j}$ is crucial, as it turns the spatial map that is obtained from merely conveying ``where the field has variance'' to instead be ``where the field has variance in the modes that are important for predictability''.
    \item Average over all $(n,l)$ combinations that lie on the diagonal $d=n+l$ for each lead $n\in\{0,\ldots,24\}$ and lag $l\in\{0,\ldots,11\}$:
    \begin{linenomath*}
    \begin{equation}
    \bar I^{(f)}_{d} \;=\; \frac{1}{N_d}\sum_{n=1}^{\min(d,24)} I^{(f)}\!\left(\mathbf x\,;l=d-n\right),
    \end{equation}
    \end{linenomath*}
    where $N_d$ is the number of contributing maps for each diagonal $d=1,\ldots,24$ months. Note: this is essentially the same diagonal averaging procedure used to compute reconstructed components in SSA \cite{Vautard1989}.
    \item Finally, rescale each field map so that it has unit mean over the ocean:
    \begin{linenomath*}
    \begin{equation}
    \bar I^{(f)}_{d} \leftarrow \frac{\bar I^{(f)}_{d}}{\left\langle \bar I^{(f)}_{d}\right\rangle_{\Omega_f}}.
    \end{equation}
    \end{linenomath*}
    For the purposes of visualisation, the $\mathrm dT/\mathrm dz$ importance maps are restricted to an equatorial average between 5$^\circ$S and 5$^\circ$N.
\end{enumerate}
This last step ensures that the three fields are on a common, dimensionless scale and makes $\bar I^{(f)}_{d}$ directly comparable across lead times and between fields, since values $>1$ indicate above-average importance within that field’s domain. Stratification by target season and/or class is also straightforward to implement as it merely consists of restricting the individual vectors $W$ that are used to form the aggregated probability measure $W^{(n)}$. Additionally, in the case of stratification by target season, the per-gridpoint standard deviations $\sigma^{(f)}(\mathbf x)$ are recomputed for only the months in that target season. 

\newpage
\begin{figure}[H] 
\centering
\includegraphics[height=\textheight]{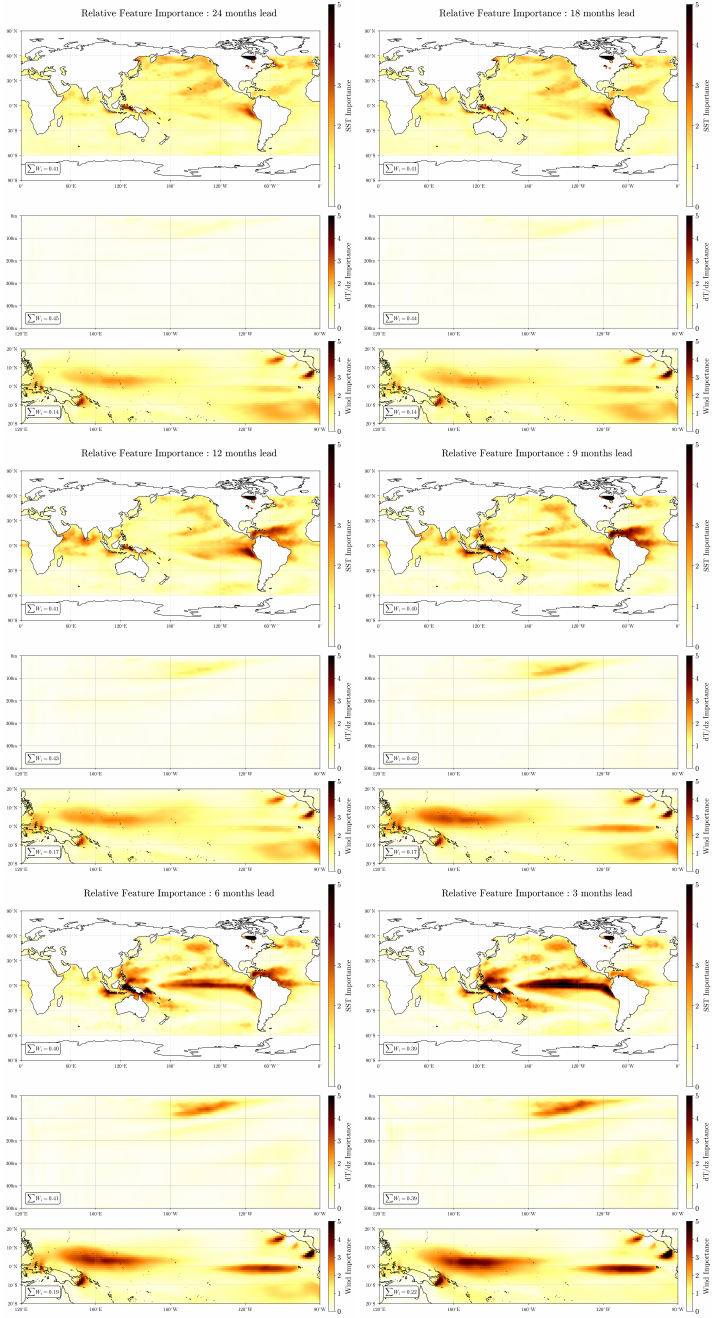}
\vspace{-1em}
\caption{Spatial feature importance maps for target season DJF at various lead times.}
\label{fig:spatial-importance}
\end{figure}

Figure \ref{fig:spatial-importance} shows an example of these spatial importance maps for lead times of 3, 6, 9, 12, 18 and 24 months (plotted in descending order), restricted to the target season DJF. While it is also possible to restrict the spatial importance maps to just the El Ni\~no/Neutral/La Ni\~na class, we find that this results in only negligible differences. To help with the comparison of relative importance between each field, the total weight in $W^{(n)}$ assigned to that field is shown in the bottom left corner of each subplot. The sequence of images shown in figure \ref{fig:spatial-importance} highlights the spatial variation in importance for predicting the phase of an ENSO event occurring in DJF at a given lead time. At the longest lead times shown (24 and 18 months), the most important regions for SST are the eastern tropical Pacific (confined mainly to near the coast of South America), the northeastern extra-tropical Pacific, the western and central Indian Ocean and some smaller hotspots located either side of New Guinea. For $\mathrm dT/\mathrm dz$, there is very little importance appearing in the equatorial slice, aside from a small region located in the upper thermocline ($z<100$m) of the central Pacific, while for wind stress the important regions are along the equator west of the dateline, as well as hotspots located off the coast of Central and South America and either side of New Guinea (i.e. in the same locations as the hotspots for SST).

For the intermediate lead times (12 and 9 months), there is a slight increase in importance in the equatorial Pacific for the SST field. At 12 months lead time, the most important regions are now the tropical Atlantic (north of the equator in particular), the Indian Ocean and the Maritime Continent, with importance remaining in the eastern equatorial Pacific off the coast of South America. At 9 months lead time, the importance in the eastern equatorial Pacific has reduced, while a pattern of importance in the northeastern extra-tropical Pacific that has been present throughout the sequence thus far, and which closely matches the Nortrh Pacific Meridional Mode \cite{Chiang2004}, has stengthened. The importance in the tropical Atlantic has also further strengthened and there is an eastward shift of importance from the Indian Ocean to the Maritime Continent. In the $\mathrm dT/\mathrm dz$ field at these lead times, the upper thermocline of the central Pacific is still the most important region, with the importance concentrating more strongly there, although the overall imnportance in the equatorial slice is still quite weak---an indication that most of the importance for this field is in the off-equatorial regions still. Meanwhile in the wind stress field, there is a westward shift in importance towards the Western Pacific Warm Pool and the Maritime Continent, coinciding with the regions where westerly wind bursts---an important trigger for El Ni\~no events at these lead times---primarily occur. There is also the emergence of a secondary region of importance in the eastern equatorial Pacific (west of the Galápagos Islands) at 9 months lead. 

For the shortest lead times (6 and 3 months), the SST importance concentrates strongly in the equatorial Pacific. At 6 months lead the tropical Atlantic is still strongly important, while the importance in the Indian Ocean has reduced. The importance over the Maritime Continent and Western Pacific Warm Pool has also increased, along with the emergence of a new region of importance over the South Pacific islands (i.e. just south of the SPCZ). At 3 months lead these regions remain important, with the main change being a shift in importance away from the tropical Atlantic and, as expected, towards the central and eastern equatorial Pacific. For $\mathrm dT/\mathrm dz$, there is only a slight eastward shift of the most important region in the upper thermocline at these lead times, however there is now a much higher concentration of importance in this region of the equatorial slice. Finally, for the wind stress field there is little change in importance between 9 and 6 months lead time aside from a reduction of importance in the South American coastal region and an increase in importance of the other areas. Similarly, between 6 and 3 months lead time, the only major change is a reduction in importance over the Maritime Continent and an eastward shift of the region of importance over the western equatorial Pacific to be concentrated closer to the dateline.

A minor caveat to note is that the current 2D equatorial slice for visualising the $\mathrm dT/\mathrm dz$ field may be missing regions of off-equatorial importance that are significant and of interest, particularly at longer lead times. However, by performing the equatorial average after the 3D field has been rescaled to have unit mean, we can at least determine whether most of the importance lies near the equator or not. Similarly, another limitation to note is that the ``true'' importance for predictability at a given lead time may actually be in other fields that are correlated with the ones we are using here. The full set of importance maps for all lead times, target seasons and classes can be viewed using the web app accompanying this article.

\subsection{Case studies} \label{subsec:case-studies}

The affiliation sequences $\Gamma^{(n)}_{:,t}$ provide a means of analysing the pathways taken by individual ENSO events through the phase space defined by the set of superclusters across all lead times. As mentioned in section \ref{subsec:superclusters}, for any particular target date $t$ in the evaluation period, the specific pathway taken to the event is determined by the affiliations $\Gamma_{:,t}^{(n)}$ at each lead $n$. Using equation \ref{eqn:Xhat}, a reconstructed feature vector $\hat X^{(n)}_{:,t}$ can be calculated at each lead time $n$ for a given target date $t$. These feature vectors can then be converted to spatiotemporal composites in the original SST, $\mathrm dT/\mathrm dz$ and wind stress fields using equation \ref{eqn:composites}.

However, because each feature dimension corresponds to an SSA mode with a 12-month embedding length, the reconstructed feature vector $\hat X^{(n)}_{:,t}$ at lead $n$ encodes information about the state of each field at lags $l=0,\ldots,11$ months prior to time $t-n$ months. As was done for the spatial importance maps in section \ref{subsec:importance}, we employ diagonal averaging to reconstruct a single time series from these embedded representations. Specifically, for each offset $d=1,\ldots,24$ months prior to the target date $t$, we average over all $(n,l)$ combinations that lie on the diagonal $d=n+l$:
\begin{linenomath*}
\begin{equation} \label{eqn:diagonal-avg}
\bar{\mathcal{C}}^{(d)}(\mathbf{x},t) = \frac{1}{N_d}\sum_{n=1}^{\min(d,24)}\mathcal{C}^{(n)}\!\left(\mathbf{x},t;l=d-n\right),
\end{equation}
\end{linenomath*}
where $N_d$ is the number of contributing $(n,l)$ pairs for each diagonal $d$ (i.e.\ those satisfying $0\le l\le 11$ and $1\le n\le 24$), and $\mathcal{C}^{(n)}(\mathbf{x},t;l)$ denotes the composite at lag $l$ produced from the reconstructed feature vector at lead $n$. Explicitly, this composite is computed by first calculating the reconstructed feature vector in the latent space using equation \ref{eqn:Xhat}, rescaling by $1/\sqrt{W^{(n)}}$ to the unweighted feature space, and then applying equation \ref{eqn:composites} at the desired lag $l$. The corresponding class probability distribution at each offset $d$ is similarly averaged:
\begin{linenomath*}
\begin{equation} \label{eqn:probs-avg}
\bar{\Pi}_{:,t}^{(d)} = \frac{1}{N_d}\sum_{n=1}^{\min(d,24)}\Lambda^{(n)}\Gamma^{(n)}_{:,t},
\end{equation}
\end{linenomath*}
where the sum is restricted to the same valid $(n,l)$ pairs.

This diagonal averaging procedure produces a sequence of 24 composite patterns for each target date, representing the reconstructed precursor states at $d=24,23,\ldots,1$ months prior to the event. Unlike the canonical pathways shown in figure \ref{fig:vertical}, which represent the most probable trajectory through the Bayesian network $\mathcal{G}$, these case study composites capture the specific combination of supercluster states that were actually occupied (in a fuzzy sense) in the lead up to a particular event at time $t$. This enables detailed comparisons between different events of the same type (e.g. different El Ni\~no events) to identify commonalities and differences in their precursor evolution. It also enables comparisons between the reconstructed precursors and observations at that same lead, providing another means of interpreting which regions the model is relying on most to successfully predict a given event.

\newpage
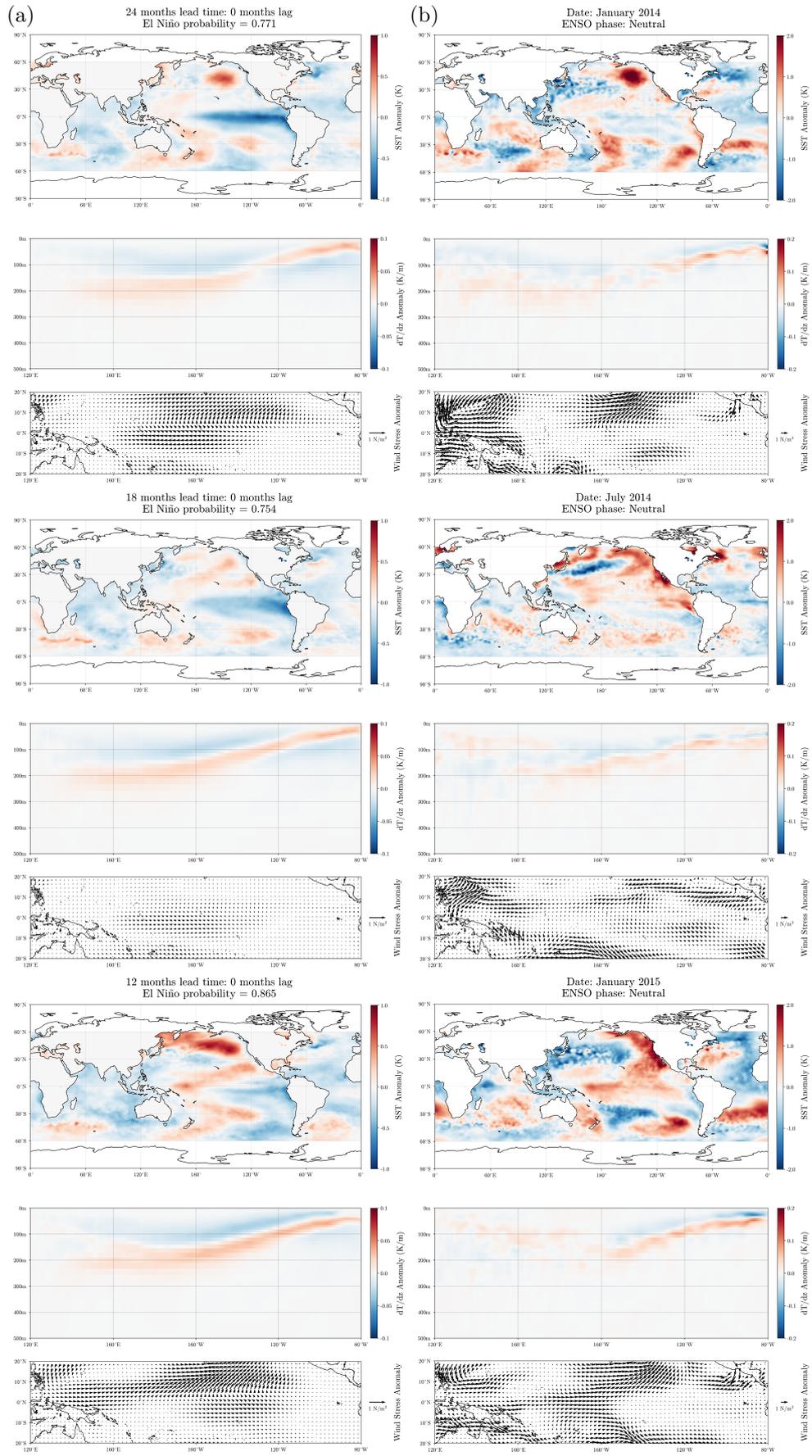
\begin{figure}[H] 
\centering
\begin{tikzpicture}
    \node (img1) {\includegraphics[height=\textheight]{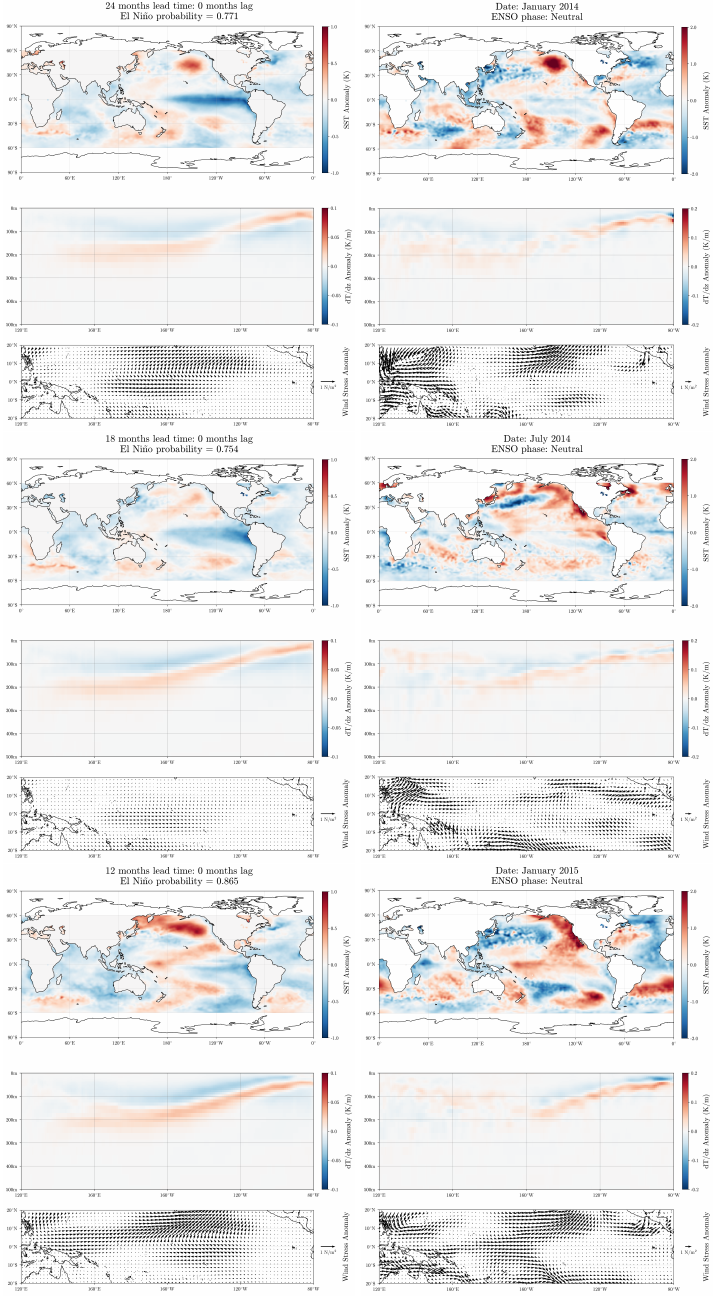}};
    \node[anchor=north west, xshift=0pt, yshift=0pt] at (img1.north west) {(a)};
    \coordinate (b) at ($ (img1.north west)!0.543!(img1.north east) $);
    \node[anchor=north east, xshift=0pt, yshift=0pt] at (b) {(b)};
\end{tikzpicture}
\vspace{-1em}
\caption{January 2016 case study, showing reconstructed (a) and observed (b) precursors.}
\label{fig:jan-2016-part-1}
\end{figure}

\newpage
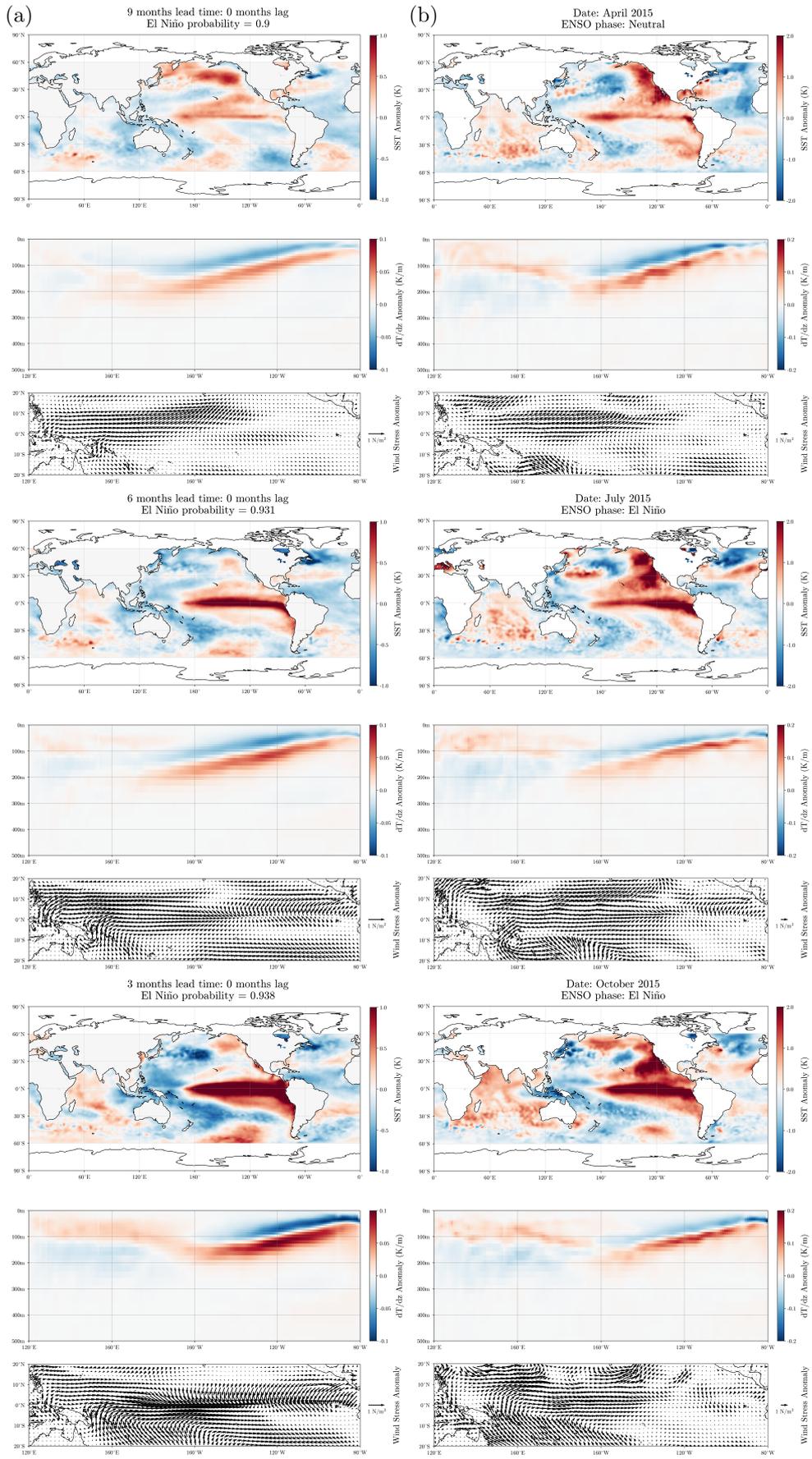
\begin{figure}[H] 
\centering
\begin{tikzpicture}
    \node (img1) {\includegraphics[height=\textheight]{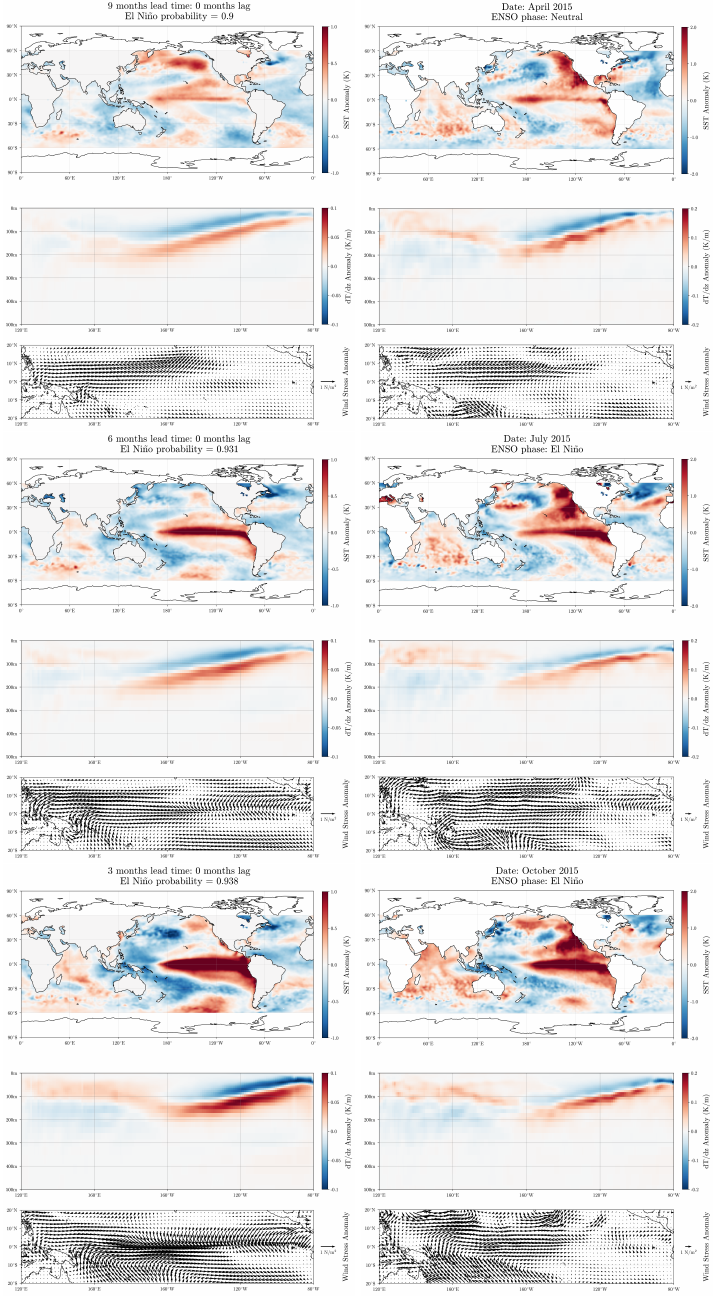}};
    \node[anchor=north west, xshift=0pt, yshift=0pt] at (img1.north west) {(a)};
    \coordinate (b) at ($ (img1.north west)!0.543!(img1.north east) $);
    \node[anchor=north east, xshift=0pt, yshift=0pt] at (b) {(b)};
\end{tikzpicture}
\vspace{-1em}
\caption{January 2016 case study, showing reconstructed (a) and observed (b) precursors.}
\label{fig:jan-2016-part-2}
\end{figure}

Figures \ref{fig:jan-2016-part-1} and \ref{fig:jan-2016-part-2} show such a comparison for a target date of January 2016, at lead times of 24, 18 and 12 months (figure \ref{fig:jan-2016-part-1}) and 9, 6 and 3 months (figure \ref{fig:jan-2016-part-2}). In both figures the left-hand side subplots show the reconstructed SST, $\mathrm dT/\mathrm dz$ and wind stress fields while the right-hand side subplots show the observed fields at the same lead prior to the target date. The titles of each left-hand side subplot contain the predicted class probability distribution (given by equation \ref{eqn:probs-avg}) at lead $n$, while the right-hand side subplot titles contain the observed phase of ENSO at that same lead. 

At 24 months lead time, the main areas of agreement in the SST field are the large warm anomaly in the northeast Pacific ocean near the Gulf of Alaska, often referred to as ``The Blob'' \cite{Bond2015}, a horseshoe-type pattern of warm anomalies that extends from the northeast Pacifc to the Maritime Continent and then back towards the south Pacific islands that is reminiscent of the negative phase of the PDO \cite{Zhang1997} as well as an anomalously cold tropical Atlantic ocean. The agreement in other regions is not as strong, with the overall area-weighted pattern correlation for SST between the reconstructed and observed fields only equal to 0.32. The agreement between the $\mathrm dT/\mathrm dz$ fields is slightly stronger, with the patterns best matching west of 160$^\circ$W (the area-weighted pattern correlation is 0.39), while the agreement between the wind stress fields is similar with an area-weighted vector pattern correlation of 0.47. 

At 18 months lead time there is still some agreement between the SST fields outside of the tropical Pacific, in particular south of 20$^\circ$S, as well as the tropical and south Atlantic, but overall the agreement is worse than at 24 months lead; as reflected by a pattern correlation of only 0.015. There is somewhat better agreement in the $\mathrm dT/\mathrm dz$ field (the pattern correlation is 0.34) and but less so for the wind stress field (vector pattern correlation of 0.14). Moving to 12 months lead time, in the SST fields the main region of agreement is the diagonal band of warm SSTs extending from the coast of North America to the equatorial Pacific---matching the canonical pattern of the North Pacific Meridional Mode---as well as a cold anomaly in the tropical Atlantic. The overall agreement is still quite low however, with a pattern correlation of 0.027. The fact that the NPMM is captured in the reconstructed composite is noteworthy as it is known to be a key trigger for El Ni\~no events via the seasonal footprinting mechanism at this lead time \cite{Vimont2001}. In terms of $\mathrm dT/\mathrm dz$ there is better agreement, mainly east of 160$^\circ$W, with a pattern correlation of 0.54. Similarly, in the wind stress field there are also now regions of good agreement (vector pattern correlation of 0.47), mainly over the same diagonal band corresponding to the NPMM pattern as well as in the western tropical Pacific, where anomalous westerlies are beginning to form. 

At 9 months lead, the regions of agreement in all three fields are mostly the same, with the exception that the reconstructed SST field now matches the warming pattern along the equator that is present in the observed SST field. The pattern correlation for SST has also improved to 0.32. For $\mathrm dT/\mathrm dz$, east of 160$^\circ$W is still the best region of agreement, although the reconstructed field now also better matches the observed field in the upper Western Pacific Warm Pool (above 100m depth) where the thermocline is shoaling and the overall pattern correlation is now 0.80. Meanwhile for wind stress the reconstructed field now closely matches the observed westerly wind bursts and the pattern correlation is 0.67. At 6 months lead time the visual agreement across all three fields improves dramatically. For SST, the reconstructed field now exhibits the characteristic warm tongue in the eastern and central equatorial Pacific and the pattern correlation has increased to 0.62. For $\mathrm dT/\mathrm dz$ the reconstructed and observed patterns match very closely (the pattern correlation is 0.82), while for wind stress the anomalous westerlies along the equator are well-captured by the reconstruction, as is the convergence zone in the eastern equatorial Pacific, with an overall pattern correlation of 0.65.  Finally, at three months lead time the reconstructed SST field has developed into the canonical El Ni\~no pattern (e.g. closely matching the SST composite for cluster 8 in figure \ref{fig:vertical}), with the area-weighted pattern correlation further improving to 0.73, while the reconstructed and observed $\mathrm dT/\mathrm dz$ fields remain well-matched (pattern correlations of 0.89 and 0.62 respectively). For wind stress the agreement is similar to 6 months lead, however the reconstructed field now matches the south-easterly anomalies in the south Pacific between 140$^\circ$E and 140$^\circ$W that correspond to a north-easterly displacement of the South Pacific Convergence Zone.

This case study, along with a similar case study for a target date of January 2018 shown in figures S9 and S10 of the supplementary material, illustrates how only a small number of key features are required to be correctly identified in order to relate the observed precursor pattern to the canonical patterns that are captured by the supercluster centroids. Moreover, these key features may only need to be present in a subset of the fields at a given lead time, and may jump between fields as the lead time varies. For the full suite of case studies, the reader is referred to the web app accompanying this manuscript. 

Under the assumption that the regions of best agreement between the reconstructed and observed fields at a given lead $n$ correspond to the features that are being used to correctly predict the phase of ENSO in $n$ months time, we can calculate the correlation between the reconstructed and observed fields at each grid point across all possible case studies in the evaluation period (i.e. target dates from January 2002 to December 2024 and lead times $n=1,\ldots,24$ for which the affiliation vector $\Gamma_{:,t}^{(n)}$ is available). These provide a complementary picture to the spatial importance maps presented in section \ref{subsec:importance}. As with the importance maps, we can restrict the correlation to just target dates for a given season, class, or both. 

Figure \ref{fig:correlation-maps} shows the correlation maps for lead times of 24, 18, 12, 9, 6 and 3 months across all target dates in the evaluation period for which an El Ni\~no was present. The range of the plots for each field has been restricted so that only positive correlations are shown (i.e. the quantity being plotted is actually $\max(r,0)$). At the longest lead times (24 and 18 months), there are only small regions in all three fields for which the correlation across all El Ni\~no events in the evaluation period is high. The correlation becomes much higher at 12 months lead; for SST the regions of highest correlation correspond to the patterns of variability associated with the North Pacific Meridional Mode and the Atlantic Meridional Mode \cite{Chiang2004}, while for $\mathrm dT/\mathrm dz$ and wind stress the regions of highest correlation correspond well with the regions of highest feature importance identified in figure \ref{fig:spatial-importance} at this lead time. The same is true for $\mathrm dT/\mathrm dz$ and wind stress at 9 months lead, while for SST there is now high correlation in the tropical Atlantic (matching well with the feature importance map at 9 months lead) and central Pacific. At 6 and 3 months lead time, the highest correlation for SST is now in the tropical Pacific where the primary warm anomaly associated with El Ni\~no occurs. For $\mathrm dT/\mathrm dz$ there is generally high correlation across the thermocline, while for wind stress the correlation in the region of westerly wind bursts increases further, along with a secondary region of high correlation west of the Galápagos Islands that agrees well with the feature importance maps.

\newpage
\begin{figure}[H] 
\centering
\includegraphics[height=\textheight]{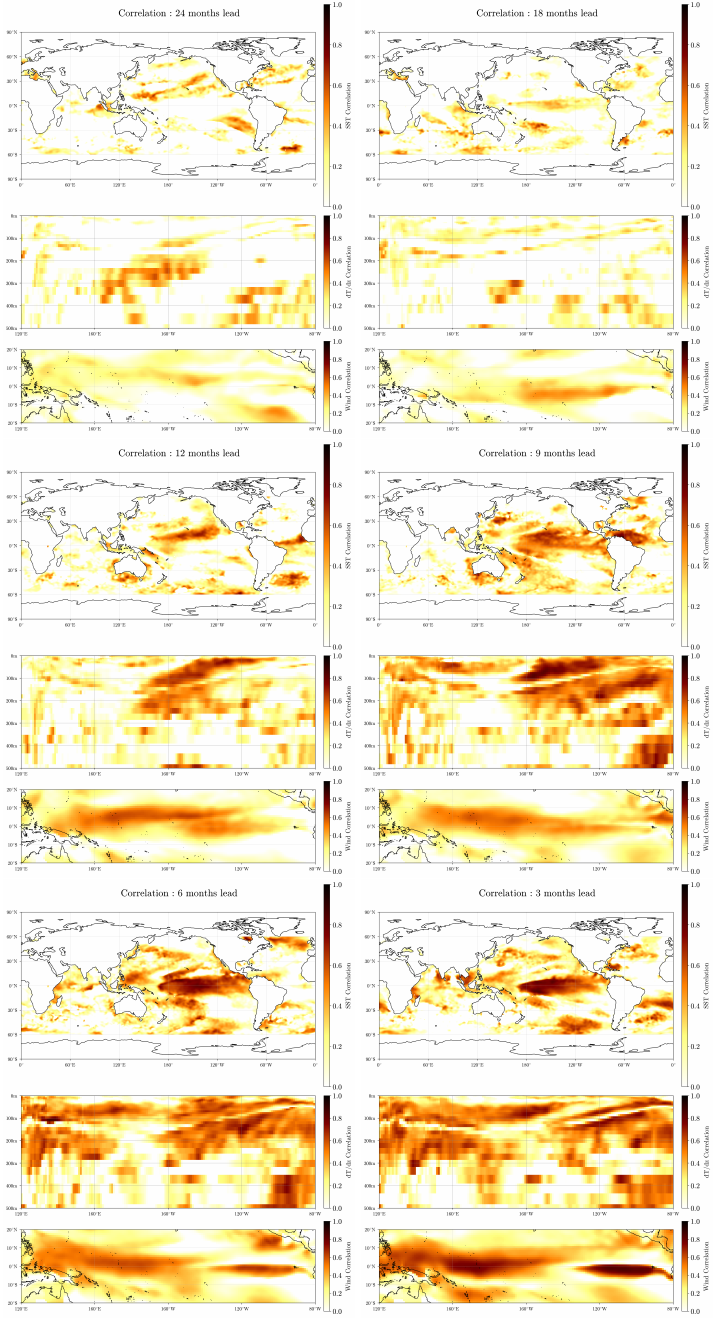}
\vspace{-1em}
\caption{Correlation maps for all target dates on which an El Ni\~no occurred.}
\label{fig:correlation-maps}
\end{figure}

Figure S11 in the supplementary material shows the same plot as figure \ref{fig:correlation-maps} but for target dates in the evaluation period for which a La Ni\~na was present. Key differences with respect to figure \ref{fig:correlation-maps} are summarised as follows:
\begin{enumerate}[label=(\roman*)]
    \item At 24 and 18 months lead there is stronger correlation for SST in the tropical Pacific and similar (weak) correlation across the $\mathrm dT/\mathrm dz$ and wind stress fields.
    \item For SST at 12 and 9 months lead there is much stronger correlation across the eastern and central Pacific. The correlation in the Indian ocean is also much higher, while in the tropical Atlantic it is lower.
    \item For $\mathrm dT/\mathrm dz$ at 12 and 9 months lead the correlation is much higher in the Western Pacific Warm Pool, lower in the central Pacific and similar in the eastern Pacific. For wind stress there is generally weaker correlation at these lead times than in the El Ni\~no case.
    \item At 6 and 3 months lead the main difference for SST is the much higher correlation in the Indian Ocean basin. For $\mathrm dT/\mathrm dz$ the correlation is high across the entire thermocline (even more so than for El Ni\~no), while for wind stress the regions of high correlation are generally similar to figure \ref{fig:correlation-maps}, with the exception of lower correlation over the Maritime Continent and in the region west of the Galápagos Islands.
\end{enumerate}

\subsection{Precursor identification plots} \label{subsec:precursors}
To conclude this section, we demonstrate how the importance maps from section \ref{subsec:importance} (or the correlation maps from section \ref{subsec:case-studies}) can be further enhanced by combining them with other composites so as to identify not just the location of a particular precursor but also its amplitude and sign. To do this, we take the (detrended, smoothed and regridded) observational data for which the SSA decomposition was computed on, form composites for a desired target season and/or observed class, and then overlay specific contours of the importance maps on top of these composites.

The procedure is as follows:
\begin{enumerate}
    \item For each target date $t$ and lead time $n$ in the evaluation period, the observational fields $F^{(f)}(\mathbf{x}, t-n)$, with $f=1,2,3$ corresponding to the SST, $\mathrm dT/\mathrm dz$ and wind stress fields respectively, are retrieved at the date $n$ months prior to the target date.
    \item The collections for each field are then aggregated to form composites, weighted by the number of valid samples. Specifically, for each lead $n=1,\ldots,24$ we compute
    \begin{linenomath*}
    \begin{equation} \label{eqn:gt-composite}
    \bar{F}^{(f)}_n(\mathbf{x}) = \frac{1}{N_n}\sum_{t\in\mathcal{T}_n} F^{(f)}(\mathbf{x}, t-n),
    \end{equation}
    \end{linenomath*}
    where $\mathcal{T}_n$ is the set of target dates for which lead-$n$ affiliations are available and $N_n = |\mathcal{T}_n|$ is the number of such dates.
    \item The composites can be stratified by target season and/or observed ENSO class by restricting the sum in equation \ref{eqn:gt-composite} to the appropriate subset of target dates. For example, to generate composites for El Ni\~no events occurring in DJF, the sum is restricted to target dates $t$ for which $\Pi_{3,t}=1$ (i.e.\ El Ni\~no class) and $\mathrm{month}(t)\in\{12,1,2\}$.
    \item The spatial importance maps $\bar{I}^{(f)}_d(\mathbf{x})$ (for $d=n$), computed for the same stratification (target season and/or class), are then overlaid as contours on top of these composites when plotting. Contour levels are chosen to highlight regions of above-average importance, specifically 1$\times$, 2$\times$ and 4$\times$ the domain-mean importance are chosen and are represented using various types of hatching (dots, diagonal lines and crosses respectively) in the case of SST and $\mathrm dT/\mathrm dz$ or as a contour flood with different shades of grey in the case of wind stress.
\end{enumerate}

\newpage
\begin{figure}[H] 
\centering
\includegraphics[height=\textheight]{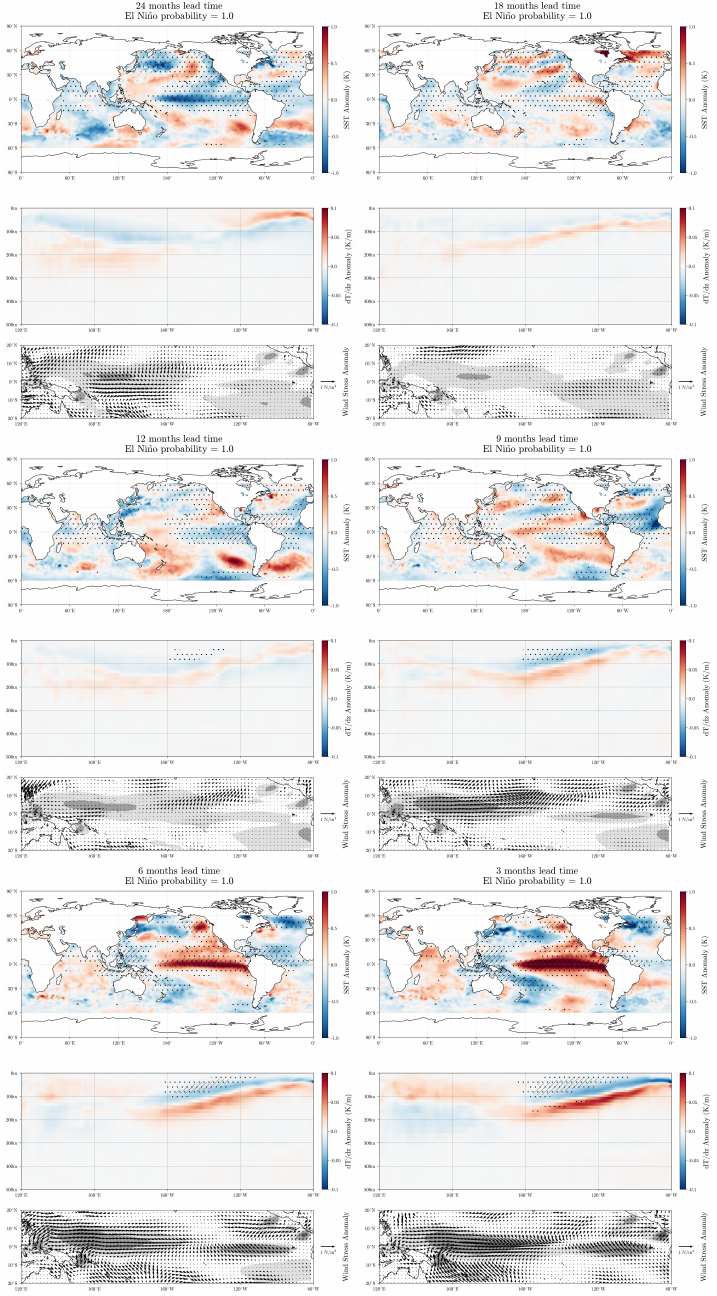}
\vspace{-1em}
\caption{Precursor plots for all DJF target dates on which an El Ni\~no occurred.} 
\label{fig:precursor-plots}
\end{figure}

This results in plots that simultaneously display: (i) the mean observed precursor pattern, indicating the sign and amplitude of anomalies at each grid point in each field, and (ii) contours delineating where the model assigns above-average importance for making successful predictions at that lead time. The combination allows for direct identification of which physical signals (precursors) the model is relying on at a given lead, helping to bridge the gap between ``where the model is looking'' and ``what does it expect to see there''. Furthermore, by stratifying by ENSO class and/or target season, differences in the precursor patterns and important regions between El Ni\~no and La Ni\~na events or different phases of the seasonal cycle can be examined.

Figure \ref{fig:precursor-plots} shows an example of this type of plot, with the composites restricted to target season DJF and the El Ni\~no class; i.e. the contours are derived from essentially the same importance maps shown in figure \ref{fig:spatial-importance} (with only negligible differences due to the additional stratification by class). Figure S12 in the supplementary material shows a similar plot for DJF and La Ni\~na. As the general trends in spatial importance for this target season have already been discussed in section \ref{subsec:importance}, the discussion here will be kept brief. The main point to note is that the regions of highest importance tend to map very nicely onto specific anomalous structures in each field, for example the cold SST anomaly in the equatorial and north-eastern Pacific at 24 months lead or the strong westerly anomalies in wind stress near 160$^\circ$E at 9, 6 and 3 months lead. This helps to further validate the feature importance methodology, confirming that the regions identified as most important for prediction correspond to genuine precursor anomalies. The consistency between the spatial importance maps and the observational composites provides additional confidence that eSPA models, conditional on predicting the correct class, are making these correct predictions by successfully capturing key dynamical pathways leading to ENSO events.

\section{Conclusions} \label{sec:conclusions}

This study has presented a comprehensive framework for distilling and interpreting the information contained within large ensembles of Entropic Learning (EL) models. While the approach detailed in \citeA{Groom2025} established that ensembles of entropy-optimal Sparse Probabilistic Approximation (eSPA) models could achieve comparable forecast skill to (multi-model ensembles of) operational dynamical systems, while also remaining skilful at lead times longer than those of the forecasts issued by such systems, the sheer volume of models in the entire ensemble presented a barrier to physical interpretation. By aggregating the internal states of successful ensemble members into a reduced set of ``superclusters'' for each lead time, we have successfully distilled the ensemble into a more parsimonious and interpretable representation while also retaining comparable probabilistic forecast skill. Furthermore, the distillation procedure is less costly to perform than the generation of the original ensemble, which was itself demonstrated in \citeA{Groom2025} to be $\sim1000-10000\times$ cheaper than a conventional dynamical model.

A key finding of this work is that the distilled superclusters capture the underlying dynamics of ENSO with remarkable fidelity. The transition probability matrices derived from supercluster affiliation sequences exhibit a mean relaxation time of 6.0 months, closely matching the observed $e$-folding time of the Ni\~no3.4 index. Furthermore, the construction of a Bayesian network linking cluster assignments across lead times revealed consistent dynamical pathways and, in combination with the transition probability matrices, provides a principled way to impute missing affiliations for target dates where no correct ensemble members exist at a given lead. These transition operators demonstrate that distilled models, as well as the original eSPA models they are based on, learn an accurate discrete approximation of the system's spatiotemporal evolution.

An analysis of the aggregated feature importance vectors provided unique insights into the information required to predict ENSO phase. We introduced the concept of (normalised) effective dimension and showed that the complexity of the input space required for correct phase prediction peaks when forecasts must traverse the boreal spring predictability barrier. This aligns with physical intuitions: predicting the onset of an event through the spring barrier requires integrating weak signals from a broad range of precursors, whereas predicting the decay of an event is a simpler task that is well-informed by the persistence of the event itself. Furthermore, the spatial importance maps derived from the aggregated feature importance vectors help identify where predictive information resides in each field and were shown to include known physical precursors, for example the North Pacific Meridional Mode. Crucially, the focus of the models was shown to shift between different regions of each field as well as across fields depending on the lead time and target season. The case studies of the 2015/16 El Ni\~no and 2017/18 La Ni\~na events further illustrate this, showing how the reconstructed fields trace the evolution from extratropical precursors (such as ``The Blob'' and PDO-like patterns) to the mature ENSO state.

Recent advances in Deep Learning (DL)-based ENSO forecasting have increasingly paired skilful architectures with post-hoc XAI (typically some type of gradient-based sensitivity/saliency analysis) to connect their predictions to physically meaningful precursors. However, these methods can result in explanations that are misleading, and many DL studies still rely heavily on CMIP-based training data (given that common approaches for circumventing the small-data problem such as data augmentation cannot be applied in this context), which further complicates the physical interpretation of the resulting attributions. In this study our goal was to explicitly learn such attributions directly from the observational record, and the approach introduced here provides a complementary and robust route to interpretability. Rather than explaining an opaque model after the fact, it yields a compact set of transparent objects (superclusters, transition operators etc) whose meaning is directly tied to the model structure. For example, the spatial importance maps play a similar role to saliency-style attribution, but are derived directly from the internal states of our distilled model rather than from gradients. Similarly, the reconstructed composites provide a concrete narrative by tracing coherent precursor-to-event pathways across lead times that can be directly compared against observations. Taken together, these results suggest that distillation of EL ensembles can match the scientific intent of XAI for DL-based forecast systems, while offering a more intrinsically interpretable diagnostic framework for long-range ENSO predictability that is directly grounded in observations.

There are several caveats to the above results that should be made explicit. Firstly, distilled quantities such as the (imputed) affiliation sequences and aggregated feature importance vectors rely on filtering ensemble members by verification, a procedure which is not available in real time. They should therefore be interpreted primarily as a diagnostic of what the subset of eSPA ensemble members that are correctly predicting ENSO phase are capturing, rather than of the ensemble as a whole. Second, to ensure a consistent latent space for interpretability, this study uses a single SSA decomposition trained on all available data. While the resulting skill of the eSPA ensemble hindcasts produced using features derived from this decomposition does not materially increase relative to a strictly real-time protocol, this choice is still a methodological trade-off that introduces a small discrepancy between the hindcasts performed in this study vs those from \citeA{Groom2025}. Finally, the interpretability products presented here are descriptive: they identify robust associations and pathways in the distilled model state space, but do not by themselves establish causality.
\newpage
There are several direct extensions that could be applied to our framework:
\begin{enumerate}
    \item Replace the verification-based filtering with a real-time model-selection or weighting scheme (e.g.\ a meta-model) so that the ``oracle'' distilled quantities can be approximated in an operational setting.
    \item Extend the affiliation sequence calculation to incorporate categorical features so that short-lead persistence information is retained without sacrificing the benefits of fuzzy cluster assignments.
    \item Calculate separate superclusters for each target season, so that the resulting affiliation sequence and transition operators become cyclostationary.
    \item Move from phase classification to continuous targets (e.g.\ relative Ni\~no3.4 index) or even multi-index formulations (e.g. ENSO + IOD), to take into account the varying amplitude between different events. This is now possible thanks to recent algorithmic developments presented in \citeA{Bassetti2025}.
    \item Expand the composites and feature importance maps to additional fields of interest (e.g. subsurface and/or off-equatorial) as well as to different datasets to test which precursors are robust across observing systems and reanalyses.
\end{enumerate}

Overall, the distillation approach presented here provides a compact and physically-interpretable representation of what an ensemble of entropic learning models is exploiting when it successfully predicts ENSO phase across lead times. By turning an unwieldy ensemble into a reduced set of prototype models, the framework enables a rigorous investigation of long-range ENSO predictability that complements real-time data-driven operational forecasts. Furthermore, the framework also complements post-hoc DL explainability methods by providing interpretability that is largely intrinsic, and can therefore be interrogated directly in terms of regimes, pathways, and physically meaningful precursors.

\appendix
\section{Threshold-count effective dimension of $W$} \label{app:effective-dimension}

Let $D\ge 2$ and let $X=(X_1,\dots,X_D)$ be i.i.d.\ draws from a continuous distribution on $\mathbb{R}$. For $\varepsilon>0$, define the softmax weights
\begin{linenomath*}
\begin{equation}
W_d(\varepsilon)
:=\frac{\exp\!\left(-X_d/\varepsilon\right)}{\sum_{j=1}^D \exp\!\left(-X_j/\varepsilon\right)},
\qquad d=1,\dots,D,
\label{eq:softmax-weights}
\end{equation}
\end{linenomath*}
and the threshold-count effective dimension
\begin{linenomath*}
\begin{equation}
\tilde D(\varepsilon)
:=\sum_{d=1}^D \mathbbm{1}\!\left\{ W_d(\varepsilon)>\frac1D\right\}.
\label{eq:deff-threshold}
\end{equation}
\end{linenomath*}

\begin{theorem}[Large-$\varepsilon$ limit of \ref{eq:deff-threshold}]
\label{thm:deff-half}
Assume that the distribution of $X_d$ is continuous and symmetric about some $\mu\in\mathbb{R}$ (i.e.\ $X_d-\mu =-(X_d-\mu)$). Then
\begin{linenomath*}
\begin{equation}
\lim_{\varepsilon\to\infty}\mathbb{E}[\tilde D(\varepsilon)]=\frac{D}{2}.
\label{eq:deff-limit}
\end{equation}
\end{linenomath*}
\end{theorem}

\begin{proof}
Consider a single realisation $x_1,\dots,x_D$ and write $\bar x:=\sum_{j=1}^D x_j\,/\,D$. Using the Taylor expansion
\begin{linenomath*}
\begin{equation} \label{eq:taylor-expansion}
\exp(-x_d/\varepsilon)=1-\frac{x_d}{\varepsilon}+\mathcal O(\varepsilon^{-2}),
\end{equation}
\end{linenomath*}
we obtain
\begin{linenomath*}
\begin{equation} \label{eq:normalizer}
\sum_{j=1}^D \exp(-x_j/\varepsilon)
= D-\frac{\sum_{j=1}^D x_j}{\varepsilon}+\mathcal O(\varepsilon^{-2})
= D\Bigl(1-\frac{\bar x}{\varepsilon}+\mathcal O(\varepsilon^{-2})\Bigr).
\end{equation}
\end{linenomath*}
\noindent Inverting \eqref{eq:normalizer} requires the following identity for a geometric series
\begin{linenomath*}
\begin{equation}
\frac{1}{1-z}=1+z+z^2+\cdots \qquad (|z|<1).
\end{equation}
\end{linenomath*}
Define $z:=\frac{\bar x}{\varepsilon}+\mathcal O(\varepsilon^{-2})$, therefore $z\to 0$ as $\varepsilon\to\infty$, and thus
\begin{linenomath*}
\begin{align}
\Bigl(\sum_{j=1}^D \exp(-x_j/\varepsilon)\Bigr)^{-1}
&=\frac{1}{D}\Bigl(1-\frac{\bar x}{\varepsilon}+\mathcal O(\varepsilon^{-2})\Bigr)^{-1} \nonumber\\
&=\frac{1}{D}\Bigl(1+z+\mathcal O(z^2)\Bigr) \nonumber\\
&=\frac{1}{D}\Bigl(1+\frac{\bar x}{\varepsilon}+\mathcal O(\varepsilon^{-2})\Bigr).
\label{eq:normalizer-inverse}
\end{align}
\end{linenomath*}
Substituting \eqref{eq:normalizer-inverse} and \eqref{eq:taylor-expansion} into \eqref{eq:softmax-weights} gives the expansion
\begin{linenomath*}
\begin{equation}
W_d(\varepsilon)
=\frac1D+\frac{\bar x-x_d}{D\,\varepsilon}+\mathcal O(\varepsilon^{-2}),
\qquad d=1,\dots,D.
\label{eq:Wd-expansion}
\end{equation}
\end{linenomath*}
Now, define $\bar X:=\sum_{j=1}^D X_j\,/\,D$. Since the random variables $X_d$ are continuous, $\mathbb{P}(X_d=\bar X)=0$. Therefore, by \eqref{eq:Wd-expansion} we have
\begin{linenomath*}
\begin{equation}
\mathbbm{1}\!\left\{W_d(\varepsilon)>\frac1D\right\}
\rightarrow
\mathbbm{1}\!\left\{X_d<\bar X\right\}
\quad\text{almost surely as }\varepsilon\to\infty.
\label{eq:indicator-limit}
\end{equation}
\end{linenomath*}
By dominated convergence (each summand is bounded by $1$, hence almost sure convergence implies convergence of expectations) and exchangeability (since $X$ is i.i.d.), we have:
\begin{linenomath*}
\begin{equation}
\lim_{\varepsilon\to\infty}\mathbb{E}[\tilde D(\varepsilon)]
=\sum_{d=1}^D \mathbb{P}(X_d<\bar X)
= D\,\mathbb{P}(X_1<\bar X).
\label{eq:deff-prob}
\end{equation}
\end{linenomath*}

It remains to show that $\mathbb{P}(X_1<\bar X)=1/2$ under symmetry. Let $Z_d:=X_d-\mu$ and $\bar Z:=\bar X-\mu$. Symmetry implies
\begin{linenomath*}
\begin{equation}
(Z_1,\dots,Z_D)\ = \ (-Z_1,\dots,-Z_D).
\end{equation}
\end{linenomath*}
The event $\{X_1<\bar X\}$ is equivalently $\{Z_1<\bar Z\}$, which is mapped under $(Z_1,\dots,Z_D)\mapsto (-Z_1,\dots,-Z_D)$ to $\{Z_1>\bar Z\}$. Hence
\begin{linenomath*}
\begin{equation} \label{eq:symmetry}
\mathbb{P}(Z_1<\bar Z)=\mathbb{P}(Z_1>\bar Z).
\end{equation}
\end{linenomath*}
By continuity, $\mathbb{P}(Z_1=\bar Z)=0$, so the two probabilities must both equal $1/2$. Substituting \eqref{eq:symmetry} into \eqref{eq:deff-prob} therefore gives \eqref{eq:deff-limit}.
\end{proof}

\begin{remark}
If the distribution of $X_d$ is continuous but not symmetric, the argument above still yields
\begin{linenomath*}
\begin{equation}
\lim_{\varepsilon\to\infty}\mathbb{E}[\tilde D(\varepsilon)]
= D\,\mathbb{P}(X_1<\bar X),
\end{equation}
\end{linenomath*}
which need not equal $D/2$. Trivially, $0\le \mathbb{P}(X_1<\bar X)\le 1$, so the limit lies in $[0,D]$.
Moreover, if $\mu:=\mathbb{E}[X_1]$ exists, then $\bar X\to \mu$ as $D\to\infty$, and under continuity of the cumulative distribution function $F$ at $\mu$ we get
\begin{equation}
\mathbb{P}(X_1<\bar X)\to \mathbb{P}(X_1<\mu)=F(\mu),
\qquad (D\to\infty),
\end{equation}
so for large $D$ the limit is asymptotically $D\,F(\mu)$.
\end{remark}

To relate the above result to eSPA, Theorem 1 of \citeA{Vecchi2022} shows that the solution for the feature importance vector $W$ is equivalent to \eqref{eq:softmax-weights} up to a constant, which can be combined with the hyperparameter $\varepsilon_E$ to give the general expansion parameter $\varepsilon$ used in the proof. $\varepsilon$ therefore controls the sparsity, or equivalently the effective dimension, of $W$, so we would like our measure of the effective dimension to vary between 1 and $D$ in the limits $\varepsilon\rightarrow 0$ and $\varepsilon\rightarrow \infty$ respectively. Theorem 1 above shows that, under some mild assumptions, this requirement is not satisfied by the threshold-count definition of the effective dimension in \eqref{eq:deff-threshold}.

%
%

\section*{Open Research Section}
The ERA5 and ORAS5 datasets are available at the following links: \url{https://cds.climate.copernicus.eu/datasets/reanalysis-era5-single-levels-monthly-means} and \url{https://cds.climate.copernicus.eu/datasets/reanalysis-oras5}. Source code for eSPA is available at \url{https://github.com/EntropicLearning/EntropicLearning.jl} 
The data used to generate the figures in this paper is available at \url{https://zenodo.org/records/18492383}.
The interactive visualisations accompanying this manuscript are available via a web app hosted at \url{https://app.michaelgroom.net/}. 

\section*{Conflict of Interest declaration}
The authors declare there are no conflicts of interest for this manuscript.

\acknowledgments
MG wishes to acknowledge Adam Malone for assistance with prototyping and implementing the online visualisations that accompany this manuscript.

%
\bibliography{references}
%

\end{document}